\documentclass[11pt]{article}
\usepackage{appendix}

\usepackage{amssymb,amsmath}

\usepackage[figuresright]{rotating}

\usepackage[applemac]{inputenc}

\usepackage{tikz}
\usepackage{pgfplots}

\usepgfplotslibrary{groupplots} 
\usetikzlibrary{pgfplots.groupplots}

\usepackage{graphicx}
\usepackage{dcolumn}
\usepackage{bm}
\usepackage{amscd,amsthm}
\usepackage{ifthen}
\usepackage{booktabs}
\usepackage{xfrac}
\usepackage{bm}
\usepackage{bbm}

\usepackage{hyperref}
\hypersetup{
    bookmarks=true,         
    unicode=false,          
    pdftoolbar=true,        
    pdfmenubar=true,        
    pdffitwindow=false,     
    pdfstartview={FitH},    
    pdfauthor={Reinhard Heckel},     
    pdfsubject={Subject},   
    pdfcreator={Reinhard Heckel},   
    pdfproducer={Producer}, 
    pdfnewwindow=true,      
    colorlinks=true,       
    linkcolor=red,          
    citecolor=green,        
    filecolor=magenta,      
    urlcolor=cyan           
}

\usepackage{enumerate}

\newtheorem{lemma}{Lemma}

\newtheorem{theorem}{Theorem}

\newtheorem{proposition}{Proposition}

\definecolor{DarkBlue}{rgb}{0,0,0.7}

\newcommand{\ab}{\va}
\newcommand{\abs}[1]{ \left| #1 \right| }
\newcommand{\twonorm}[1]{\left\|#1\right\|_{2}}
\newcommand{\onenorm}[1]{\left\|#1\right\|_{1}}
\newcommand{\infnorm}[1]{\left\|#1\right\|_{\infty}}

\newcommand{\opnorm}[1]{\left\|#1\right\|}

\newcommand{\val}{ {\bm \alpha} }

\newcommand{\T}{\mathcal S} 

\newcommand{\K}{{K}}

\newcommand{\derd}{\partial}

\usepackage{mathtools}

\usepackage{framed}

\newcommand{\Tnorm}{}
\newcommand\norm[2][\Tnorm]{\ensuremath{{\left\|#2\right\|}_{#1}}}

\newcommand\Tinnerprod{}
\newcommand{\innerprod}[3][\Tinnerprod]{\ifthenelse{\equal{#1}{}}{\ensuremath{\left<#2,#3\right>}}{\ensuremath{\left<#2,#3\right>_{#1}}}}

\newcommand\Tex{}
\newcommand\PR[2][\Tex]{
\ifthenelse{\equal{#1}{}}{{\mathbb{P}}\left[#2\right]}{\ensuremath{{\mathbb{P}}_{#1}\left[ #2\right]}}}
\newcommand\EX[2][\Tex]{
\ifthenelse{\equal{#1}{}}{{\mathbb E}\left[#2\right]}{\ensuremath{{\mathbb E}_{#1}\left[ #2\right]}}}

\newcommand\Var[2][\Tex]{
\ifthenelse{\equal{#1}{}}{{\mathrm Var}\left[#2\right]}{\ensuremath{{\mathrm Var}_{#1}\left[ #2\right]}}}

\newcommand{\inv}[1]{  {#1}^{ -1 } } 
\newcommand{\herm}[1]{{#1}^H} 
\newcommand{\transp}[1]{{#1}^T} 

\newcommand\defeq{\coloneqq}

\newcommand\vect[1]{\mathbf #1}

\newcommand{\va}{\vect{a}}

\newcommand{\ve}{\vect{e}}
\newcommand{\vf}{\vect{f}}  
\newcommand{\vg}{\vect{g}}  
\newcommand{\vh}{\vect{h}}

\newcommand{\vm}{\vect{m}}  
\newcommand{\vn}{\vect{n}}
  
\newcommand{\vp}{\vect{p}}  
\newcommand{\vq}{\vect{q}}
\newcommand{\vr}{\vect{r}}  
\newcommand{\vs}{\vect{s}}  

\newcommand{\vu}{\vect{u}}  
\newcommand{\vv}{\vect{v}}  
\newcommand{\vw}{\vect{w}}
\newcommand{\vx}{\vect{x}}  
\newcommand{\vy}{\vect{y}}  
\newcommand{\vz}{\vect{z}}

\newcommand{\mA}{\vect{A}}  
\newcommand{\mB}{\vect{B}} 
 
\newcommand{\mD}{\vect{D}}

\newcommand{\mF}{\vect{F}}
\newcommand{\mG}{\vect{G}}
\newcommand{\mH}{\vect{H}}
\newcommand{\mI}{\vect{I}}

\newcommand{\mL}{\vect{L}}

\newcommand{\mR}{\vect{R}}

\newcommand{\mU}{\vect{U}}
\newcommand{\mV}{\vect{V}}

\renewcommand{\S}{\mathcal S}

 \newcommand{\complexset}{\mathbb C}
 \newcommand{\reals}{\mathbb R}

   \newcommand{\mc}{\mathcal}

   \newcommand\comp[1]{ {#1}^c}

\newcommand{\setA}{\mathcal A}

\renewcommand{\L}{L}
\newcommand{\M}{R}
\newcommand{\N}{N}
\newcommand{\R}{M}

\renewcommand{\S}{S}

\newcommand{\FK}{F} 

\newcommand{\sign}{\mathrm{sign}}

\newcommand{\infdist}[1]{\left| #1\right|}

\renewcommand\Re[1]{\mathrm{Re}(#1)}
\renewcommand\Im[1]{\mathrm{Im}(#1)}

\renewcommand{\d}{d}

\newcommand{\minlet}[1]{\tilde #1} 

\newcommand\SRF{\mathrm{SRF}}

\newcommand\ind[1]{\mathbbm{1}{\{#1\}}}

\newcommand{\vct}[1]{\bm{#1}}
\newcommand{\mtx}[1]{\bm{#1}}
\newcommand{\C}{\mathbb{C}}

\newcommand\dist{F}
\newcommand\vectbr[1]{[#1]} 

\long\def\comment#1{}

\title{
Generalized Line Spectral Estimation via Convex Optimization
}

\author{
Reinhard Heckel\thanks{Department of Electrical Engineering and Computer Sciences, University of California, Berkeley, CA}  
\, and 
 Mahdi Soltanolkotabi\thanks{Ming Hsieh Department of Electrical Engineering, University of Southern California, Los Angeles, CA 
 }  
}

\usepackage[
url=false,isbn=false,doi=false,
backend=bibtex,
url=false,isbn=false,doi=false,
hyperref=auto,
style=alphabetic,
natbib=true,
sorting=nty,
firstinits=true,
maxbibnames=10,
]{biblatex}

\newbibmacro*{journal}{%
  \iffieldundef{journaltitle}
    {}
    {\printtext[journaltitle]{%
       \printfield[myplain]{journaltitle}%
       \setunit{\subtitlepunct}%
       \printfield[myplain]{journalsubtitle}}}}
\DeclareFieldFormat[article,inproceedings,incollection]{titlecase}{\MakeSentenceCase{#1}}

\bibliography{bibliography.bib} 

\newcommand{\cfar}{0.18} 

\newcommand{\mAA}{\mA}

\newcommand\indstart{1}
\newcommand\indsecond{2}
\newcommand\indend[1]{{#1}}

\newcommand\rindstart{0}
\newcommand\rindsecond{1}
\newcommand\rindend[1]{{#1-1}}

\begin{document}

\maketitle

\begin{abstract}
\emph{Line spectral estimation} is the problem of recovering the frequencies and amplitudes
of a mixture of a few sinusoids from equispaced samples. However, in a variety of signal processing problems arising in imaging, radar, and localization we do not have access directly to such equispaced samples. 
Rather we only observe a severely undersampled version of these observations through linear measurements. 
This paper is about such \emph{generalized line spectral estimation} problems. We reformulate these problems as sparse signal recovery problems over a \emph{continuously indexed} dictionary which can be solved via a convex program. We prove that the frequencies and amplitudes of the components of the mixture can be recovered perfectly from a near-minimal number of observations via this convex program. This result holds provided the frequencies are sufficiently separated, and the linear measurements obey natural conditions that are satisfied in a variety of applications.
\end{abstract}



\section{Introduction}
Many signals of interest in practical scenarios 
can be decomposed into short sums of simple objects. 
Signals arising in imaging applications, radar, and direction of arrival estimation  
are a few examples. 
Often, these signals can be represented as a sparse linear combination of elements in a \emph{continuously indexed} dictionary, with the index corresponding for example to the location of an object. 
Such sparse representations typically allow these signals to be recovered from very few observations. 

Traditional approaches to recover signals that are sparse in a continuously dictionary are spectral estimation techniques like Prony's method, MUSIC, and ESPRIT \cite{stoica_spectral_2005}. However, these techniques are only applicable to a few specific continuously indexed dictionaries. 

A different line of research in the area of compressive sensing has focused on recovery and denoising of signals that are sparse in a discretely indexed dictionary via greedy methods \cite{tropp_greed_2004} and convex programs~\cite{candes_robust_2006}. In principle, these approaches can also be applied to the recovery of signals that are sparse in continuously indexed dictionaries by discretizing the continuous parameter space. However, this discretization step induces a gridding error. While in practice---provided the grid is chosen sufficiently fine---the gridding error is negligible, fine discretization leads to dictionaries with extremely correlated, i.e.,  coherent, columns. 
Most of the theory of compressive sensing relies on the dictionary to be incoherent and therefore does not apply to fine grids.

In recent advances convex programming techniques have been shown to succeed for signals that have a sparse representation in continuously indexed dictionaries. Specifically, in~\cite{candes_towards_2014,candes_super-resolution_2013,fernandez-granda_super-resolution_2015} it is shown that the frequencies of a weighted superposition of complex sinusoids can be recovered perfectly from evenly spaced samples by solving a convex total-variation norm minimization program. This result holds provided that the frequencies of the sinusoids in the mixture are sufficiently separated. A similar statement continues to hold when only a random subset of the uniformly spaced samples are available~\cite{tang_compressed_2013}. 
Related convex programs have been studied for denoising~\cite{bhaskar_atomic_2012}, 
signal recovery from short-time Fourier measurements~\cite{aubel_theory_2015}, deducing time-frequency components of a signal with applications in super-resolution radar~\cite{heckel_super-resolution_2015}, and line spectral estimation from multiple measurements~\cite{li_off--grid_2016}. 

The goal of this paper is to characterize an even larger class of continuously indexed dictionaries, for which such convex programs provably succeed. This class includes several practically relevant dictionaries, such as 
those arising in super-resolution radar and multiple-input multiple-output (MIMO) radar. While our analytical results are for the noiseless case, we provide simulations showing that our approach is robust to noise. 
We want to highlight that a major advantage of the convex optimization based approach we use in this paper is its wide and straightforward applicability. Indeed, for a number of applications such as super-resolution MIMO radar, we are not aware of any other algorithm that provably succeeds under mild conditions.


\subsection{The generalized line spectral estimation problem}
In this section we formally introduce the problem considered in this paper. To this aim, consider a discrete vector $\vz \in \complexset^{2\N+1}$ consisting of equispaced sampled of a mixture of complex sinusoids
\begin{align}
\label{eq:defz}
\vz = \sum_{k=1}^{\S} b_k \vf(\nu_k),
\quad\text{with}\quad\vf(\nu)
\defeq \transp{[e^{-i2\pi \N\nu},e^{-i2\pi (\N-1)\nu},\ldots,  e^{i2\pi \N\nu}]}.
\end{align}
Here, $\{b_k\}$ are complex-valued  coefficients, $\{\nu_k\}$ are frequency parameters that can take any value in $[0,1]$, and $\S$ is the (unknown) number of components in the mixture. 
In this paper, the goal is to deduce these coefficients and frequency parameters from linear observations of this mixture. 
More specifically, we consider the problem of recovering $\{b_k\}$ and $\{\nu_k\}$ from measurements of the form
\begin{align}
\label{eq:syseqinw}
\vy = \mAA \vz \in \complexset^\R, 
\end{align}
where $\mAA \in \complexset^{\R \times \L}$ is a sensing matrix and $\L \defeq 2\N +1$. 
The focus of this paper is on the regime where $\R < \L$; thus $\mAA$ is not invertible. If $\mAA$ were invertible, recovery of the unknowns from $\vy$ reduces to recovery from $\vz$. 
The latter is known as a \emph{line spectral estimation} problem. 
We thus refer to the more general problem of recovering the parameters $\{(b_k,\nu_k)\}$ from $\vy = \mA \vz$ as the \emph{generalized line spectral estimation} problem. 

Line spectral estimation is a classical signal processing problem that can be efficiently solved with standard techniques such as Prony's method, MUSIC, and ESPRIT \cite{stoica_spectral_2005}. 
However, these standard spectral estimation techniques are in general not applicable to the generalized line spectral estimation problem, in particular not when $\R < \L$. Recently, there has been a surge of interest in using ideas based on convex optimization for solving the classical line spectral estimation problem and variants thereof. Specifically,~\cite{candes_towards_2014} show that a convex optimization scheme provably recovers the mixture coefficients and frequencies from $\vz$ provided the frequencies are sufficiently separated. As mentioned before, related convex programs have also been studied for diverse problems arising in applications from denoising to radar \cite{tang_compressed_2013, bhaskar_atomic_2012, aubel_theory_2015, heckel_super-resolution_2015,cai_robust_2016}. All these problems are special instances of the generalized line spectral estimation problem 
or the higher dimensional versions thereoff, 
each corresponding to a specific matrix $\mA$ in \eqref{eq:syseqinw}.  
We wish to emphasize that in particular many relevant problems in radar and localization correspond to the high dimensional version (see Equation~\eqref{measA} in Section~\ref{sec:higherdim})
of the one-dimensional relation \eqref{eq:syseqinw}. 

Note that the generalized line spectral estimation problem is a sparse signal recovery problem: 
due to $\vy = \sum_{k=1}^{\S} b_k \mA \vf(\nu_k)$, recovery of the parameters $\{(b_k,\nu_k)\}$ corresponds to the recovery of a signal that is sparse in the continuously indexed dictionary $\{\mA \vf(\nu),  \nu \in [0,1]\}$. 
As it turns out, many practically relevant \emph{continuously} indexed dictionaries that occur in applications such as localization and radar are of the form $\{\mA \vf(\nu),  \nu \in [0,1]\}$. 

The main goal of this paper is to show that for a large class of measurement matrices $\mA$, 
and therefore for a large class of continuously indexed dictionaries, 
a convex program provably recovers the coefficients $\{b_k\}$ and frequencies $\{\nu_k\}$ from the measurements $\vy$. Our results encompass those in \cite{tang_compressed_2013,aubel_theory_2015,heckel_super-resolution_2015}, but allow for many new measurement models which are of interest in diverse applications such as MIMO radar.


\subsection{What mixtures?}

Throughout we assume that the frequencies of the components of the mixture are sufficiently separated.
Specifically, we require that
\begin{align}
\label{eq:minsepintro1D}
\abs{\nu_k - \nu_{k'} } \geq \frac{2}{\N}, 
\quad
\text{ for all } k, k' \in \{\indstart,\ldots,\indend{\S}\},  k\neq k'.
\end{align}
Here, and throughout, $|\nu_k - \nu_{k'}|$ is the wrap-around distance on the unit circle. For example, $|3/4-1/2|=1/4$ but $|5/6-1/6|=1/3\neq 2/3$. 
We note that some form of separation between the frequencies is necessary for \emph{stable} recovery.
This follows from identification of the mixture components being extremely ill conditioned when the frequencies are clustered closely together, even when $\mA = \mI$. 
Specifically, if there are $\S'$ components with frequencies $\nu_k$ in an interval of length smaller than $\frac{2S'}{L}$, and $\S'$ is large, then in the presence of even a tiny amount of noise, stable recovery from $\vz$ is not possible (e.g.~see \cite[Thm.~1.1]{donoho1992superresolution} and \cite[Sec.~1.7]{candes_towards_2014}). Condition~\eqref{eq:minsepintro1D} 
allows us to have  $0.5\,\S'$ time-frequency shifts in an interval of length $\frac{2\S'}{\L}$, which is optimal up to the constant $0.5$.

\subsection{What measurements?} 
\label{whatmeas}
From \eqref{eq:defz} we know that the signal $\vz$ can be written as a sparse linear combination of at most $S$ atoms $\vf(\nu)$ in the set \mbox{$\setA \defeq \{ \vf(\nu), \nu \in [0,1]\}$}. Clearly, one can not hope to solve the generalized line spectral estimation for all measurement matrices even when there are many measurements. For instance, if all the measurement vectors (rows of $\mA$) are orthogonal to the elements of $\setA$ then these measurements contain no information about the signal and one can not possibly hope to do spectral estimation. 

We next describe a large class of measurement matrices $\mA$ that allow to recover the parameters $\{(b_k,\nu_k)\}$ from $\vy$. 
For $r=\rindstart,\rindsecond,\ldots,\rindend{\R}$, let $\va_r\in\C^\L$ denote the rows of $\mA$ (a.k.a.~the measurement vectors). We assume that the measurement vectors are drawn i.i.d.~according to a distribution $\dist$ obeying two natural conditions, specifically the isotropy and incoherence properties discussed below. 
We note that we do not actually require the rows of the measurement matrix $\mA$ to be i.i.d. 
Indeed, our results continue to hold if the empirical distribution of the rows of $\mA$ obeys the isotropy property and incoherence properties stated below. 
It is therefore sufficient if $\dist$ corresponds to a distribution that picks a vector $\va\in\complexset^\L$ uniformly at random from the rows of $\mA$.

\paragraph{Isotropy property:} The distribution $\dist$ obeys the isotropy property if the measurement vectors $\va\sim F$ obey
\begin{align}
\mathbb{E}_{\va\sim F}\big[\va\herm{\va}\big] = \frac{1}{R} \mI.
\label{eq:isotropy}
\end{align}
This condition ensures diversity in the measurements. Indeed, if all the measurement vectors are skewed towards certain directions, i.e.,~$\va_r\approx \va$ for some vector $\va$ then an increase in the number of measurements does not reveal additional information about the signal. In this paper we assume that the isotropy condition holds. We note however that our results hold more generally with the identity matrix replaced with any invertible covariance matrix.  


\paragraph{Incoherence property:} The distribution $\dist$ obeys the incoherence property with parameter $\mu$ if for all $\va\sim F$ we have
\begin{align}
\sup_{\vf \in \complexset^\L\colon \infnorm{\vf} \leq 1} \left|\innerprod{\vf}{\va}\right|^2
=\|\va\|_{\ell_1}^2\leq \frac{\L}{\R} \mu.
\label{eq:incoherenceintro1}
\end{align}
Small values of $\mu$ imply that the sensing vector $\va$ is incoherent with all atoms in the set 
\[
\setA = \{ \vf(\nu), \nu \in [0,1]\} \subset \{ \vf \in \complexset^\L\colon \infnorm{\vf} \leq 1 \}. 
\] 
We note that the incoherence condition can be relaxed in the sense that \eqref{eq:incoherenceintro1} does not need to hold deterministically and it is sufficient that it holds with high probability. In Section~\ref{sec:JLintro}, we discuss sensing matrices that are incoherent in a probabilistic sense. We note that $\mu\ge 1$ as by the isotropy condition we have
\[
\EX{
\sup_{\vf \in \complexset^\L\colon \infnorm{\vf} \leq 1} \left|\innerprod{\vf}{\va}\right|^2
}
=
\EX{\onenorm{\va}^2} 
\geq 
\EX{\twonorm{\va}^2}
=\text{trace}\left(\EX{\va\herm{\va}}\right) =\frac{\text{trace}(\mtx{I})}{\R} =\frac{\L}{\R},
\]
 which in turn implies $\mu \geq 1$.

Our first result (Theorem \ref{thm:incoherence} stated in Section \ref{sec:mainres}) guarantees that with high probability, the vector $\vz$ (and in fact the coefficient/frequency pairs $\{(b_k,\nu_k)\}$ of \eqref{eq:defz}) can be recovered from $\R$ linear measurements of the form $\vy=\mA\vz\in\mathbb{C}^\R$ by solving a convex program. This result holds provided that the minimum separation condition \eqref{eq:minsepintro1D} holds, and provided that the empirical distribution $\dist$ of the rows of the measurement matrix $\mA$ obeys the isotropy and incoherence property with number of measurements satisfying
\begin{align*}
\R \geq c \mu  \S\log^2(\L),
\end{align*}
with $c$ a fixed numerical constant.

Our result states that the more incoherent the sensing vectors are to the atoms in $\mc A$, the fewer samples we need for signal recovery. 
This results is optimal up to logarithmic factors, since the minimum number of measurements required for any method to succeed is at least on the order of $\S$. 
Our result parallels well known results pertaining to $\ell_1$-minimization for signal recovery for the case where $\mc A$ is a basis for $\complexset^{\L}$, see \cite{candes_robust_2006, candes_probabilistic_2011}, and \cite[Ch.~12]{foucart_mathematical_2013}. 
Remarkable about our result is that the intuition developed in   \cite{candes_robust_2006,candes_probabilistic_2011} carries over, even though $\setA$ is not a basis and contains infinitely many points. 
The only additional condition that we need compared to  \cite{candes_robust_2006,candes_probabilistic_2011}, \cite[Ch.~12]{foucart_mathematical_2013} 
is the minimum separation condition. We next discuss three particular classes of measurement modalities which obey the isotropy and incoherence properties.

\subsubsection{\label{sec:subsamorthog} Subsampled orthogonal matrices}
Let $\mU$ be an orthogonal basis obeying $\herm{\mU} \mU = \frac{\L}{\R} \, \mI$ and let $\dist$ be the distribution that drawn the sensing vectors independently and uniformly at random from the rows of $\mU$. In this case the distribution $\dist$ is isotropic and incoherent with incoherence parameter $\mu = \max_k  \frac{\R}{\L} \onenorm{\vu_k}^2$, where $\vu_k$ is the $k$-th row of $\mU$.  As a more specific example, assume $\mU = \sqrt{\frac{\L}{\R}} \mI$, where $\mI$ is the identity matrix. 
In this case, $\mu = 1$. 
This corresponds to the ``compressive sensing off the grid'' \cite{tang_compressed_2013} measurement model, which is the continuous version of the classic compressive sensing problem studied in the seminal work of Cand\`es, Romberg and Tao \cite{candes_robust_2006}. 

\subsubsection{Random sampling of sparse trigonometric polynomials}

Let us consider another example, not covered by previous results in the literature. 
Assume we have a continuous signal $z(t) = \sum_{k=\indstart}^{\indend{\S}} b_k e^{i2\pi \nu_k t }$, and suppose we take $\R$ independent measurement of $z$ by sampling $z$ at locations $t_r$ chosen uniformly at random in the interval $[0,\L-1]$. 
This is equivalent to sampling from a distribution $\dist$ where the $r$-th sensing vector is equal to
\[
\va_r = \frac{1}{\sqrt{\R}}
\herm{\mF} \vf(t_r/\L),
\]
with $t_r$ drawn independently and uniformly from $[0,\L-1]$. 
Here, $\mF\in \complexset^{\L\times \L}$ is the Discrete Fourier Transform (DFT) matrix with entries $[\mF]_{p,\ell} = \frac{1}{\sqrt{\L}} e^{-i2\pi p\ell/\L}$, and 
the entries of $\vf(t_r/\L)$ are given by $[\vf(t_r/\L)]_{p} = e^{i2\pi p \frac{t_r}{\L} }$ with $p=-\N,\ldots,\N$. 
Note that
\[
\EX{\va_r \herm{\va}_r} 
= \frac{1}{\R} \herm{\mF} \EX{\vf(t_r/\L) \herm{\vf}(t_r/\L) } \mF 
= \frac{1}{\R} \herm{\mF} \mF
=\frac{1}{\R} \mI,
\]
so that the distribution $\dist$ is isotropic. 
Moreover, the distribution is incoherent with coherence parameter $c \log^2(\L)$, where $c$ is a numerical constant. 
To see that latter, first note that
\[
\onenorm{\va_r} 
= \sqrt\frac{\L}{\R} 
\sum_{\ell = -\N}^\N
\left| \frac{1}{L} \sum_{p = -\N}^\N e^{-i2\pi p \frac{\ell - t_r}{\L}} \right|.
\]
The term in the absolute value is the Dirichlet kernel evaluated at $\frac{\ell - t_r}{\L}$. 
Since the $L_1$-norm of the Dirichlet kernel is upper bounded by $\tilde c \log(\L)$, the incoherence parameter can be bounded by $\mu = c \log^2(\L)$. It then follows from our result (i.e.~Theorem \ref{thm:incoherence}) that, with high probability, the frequencies $\nu_k$ can be perfectly recovered from the samples $y_r=z(t_r)$ for $r=\rindstart,\ldots,\rindend{\R}$, provided that 
$
\R \geq \S c' \log^4(\L)
$,
and the minimum separation condition holds. 
This generalizes the main result in \cite{rauhut_random_2007}, which assumes the frequencies $\nu_k$ lie on a grid with spacing $1/\L$. 

%
%


\subsubsection{\label{sec:JLintro}Random projection matrices}

We next note that our results also hold for a large class of random projection matrices, 
including matrices with i.i.d.~sub-Gaussian\footnote{A random variable $x$ is sub-Gaussian \cite[Sec.~7.4]{foucart_mathematical_2013} if its tail probability satisfies $\PR{|x| >t} \leq c_1 e^{-c_2 t^2}$ for constants $c_1, c_2>0$. 
} entries \cite[Lem.~9.8]{foucart_mathematical_2013}. For instance, matrices with i.i.d.~Gaussian entries and i.i.d.~entries that are uniformly distributed on $\{-1/\sqrt{\R}, 1/\sqrt{\R} \}$. 
While those matrices are not deterministically incoherent in the sense of of definition~\eqref{eq:incoherence}, 
they are in a probabilistic sense incoherent with respect to \emph{any} vector and in particular with respect to the vectors $\{\vf \in \complexset^\L\colon \infnorm{\vf} \leq 1 \}$. 

We can also guarantee exact signal recovery when using such matrices (see Theorem~\ref{thm:ISG} in Section \ref{sec:mainres} for a formal statement). 
Specifically, our result guarantee that with high probability, the pairs $\{(b_k,\nu_k)\}$ and in turn the signal $\vz$ in~\eqref{eq:defz} can be perfectly recovered by a convex program, 
provided that the minimum separation condition \eqref{eq:minsepintro1D} holds and the number of measurements satisfies
\[
\R \geq c \S \log\left(\L\right),
\]
where $c$ is a fixed numerical constant. 

\subsection{\label{sec:higherdim}Extensions to higher dimensions}
The results in Section~\ref{whatmeas} generalize to higher dimensions. 
This is particularly important, as many applications are two or three dimensional generalized line spectral estimation problems. 
Specifically, consider the signal
\begin{align}
\vz = \sum_{k=\indstart}^{\indend{\S}} b_k \vf(\vr_k),
\label{eq:syseqinwdim}
\end{align}
where the complex valued coefficients $\{b_k\}$ and the vectors $\{\vr_k\}$ in $[0,1]^\d$ containing  continuous frequencies are the parameters we wish to estimate.  
Here, the entries of $\vf(\vr)$ are given by $[\vf(\vr)]_{\vp} = e^{i2\pi \innerprod{\vp}{\vr} }$ where $\vp$ an integer vector with $\ell$-th entry, $\ell=1,\ldots,\d$, given by $[\vp]_\ell=-\N,\ldots,\N$. 
We observe measurements of the form 
\begin{align}
\label{measA}
\vy = \mAA \vz\in\mathbb{C}^\R,
\end{align}
 where $\mAA \in \complexset^{\R \times \L^\d}$ is the measurement matrix with $\L = 2\N +1$, as before. 
In the $\d$-dimensional case, we require the minimum separation condition
\begin{align}
\label{eq:minsephdim}
\max_{\ell=1,\ldots,\d} |[\vr_k]_\ell - [\vr_{k'}]_\ell | \geq \frac{c(d)}{\N}, 
\end{align}
with $c(d)$ a constant only depending on the dimension $d$, to hold. 
We note that the minimum separation condition~\eqref{eq:minsephdim} only requires the $\vr_k$ to be separated in one coordinate, but not in all coordinates simultaneously. Motivated by applications such as MIMO radar in this paper we only focus on $\d\leq 3$. 
In these cases one can assume $c(d)=5$. However, our results easily extend to higher dimensions ($d>3$). A variety of signal processing problems arising in imaging, radar, and localization can be phrased as a generalized line spectral estimation problem in higher dimensions. 
We next discuss two particular radar examples.

\subsubsection{\label{sec:superres}Super-resolution radar}
A radar aims to identify the relative velocities and distances of moving targets, often modeled as point scatterers. 
To this end, the radar transmits a suitably chosen probing signal, and in response receives a weighted superposition of delayed and Doppler shifted versions of the probing signal. 
Estimating the delays and Doppler shifts corresponding to the different targets from the response signal at the receiver yields the relative distances and velocities of the targets. 
As shown in \cite{heckel_super-resolution_2015}, the response to a band-limited probing signal observed within a certain time interval is described by the input-output relation
\begin{align}
\label{eq:tfshiftssuperresradar}
\vy = \sum_{k=\indstart}^{\indend{\S}} b_k \mc F_{\nu_k} \mc T_{\tau_k} \vx,
\end{align}
where $\vx \in \complexset^\L$ is the probing signal  
and $\mc F_\nu$ and $\mc T_\tau$ are time and frequency shift operators\footnote{The time and frequency shift operators are formally defined as $[\mc F_\nu \vx_p]_p = x_p e^{i2\pi p \nu}$ and $\mc T_\tau = \herm{\mF} \mc F_\tau \mF$ where $\mF$ is a one dimensional Discrete Fourier Transform (DFT) matrix, with the entry in the $k$-th row and $r$-th column given by $[\mF]_{k, r} \defeq \frac{1}{\sqrt{\L}} e^{-i2\pi \frac{kr}{\L}}, k,r = 0,\ldots, \L-1$.}.
As shown in \cite[Sec.~6]{heckel_super-resolution_2015}
\begin{align}
\label{eq:tfshiftGF}
\mc F_{\nu}  \mc T_{\tau} \vx 
=
\mG_{\vx} \herm{\mF} 
\vf(\vr), \quad \vr = \transp{[\tau,\nu]},
\end{align}
where $\mG_{\vx} \in \complexset^{\L \times \L^2}$ is the Gabor matrix defined by 
\begin{align}
\label{eq:defgabormtx}
[\mG_{\vx}]_{p, (k,\ell)} \defeq x_{p- \ell}  e^{i2\pi \frac{k p}{\L}}, \quad k,\ell, p = -\N,\ldots,\N, 
\end{align}
and $\herm{\mF}$ is the inverse 2D discrete Fourier transform matrix with the entry in the $(k,\ell)$-th row and $(r,q)$-th column given by $[\herm{\mF}]_{(k,\ell), (r,q)} \defeq \frac{1}{\L^2} e^{i2\pi \frac{qk + r\ell}{\L}}$. 
Note that the entries $x_\ell$ of the vector $\vx$ in \eqref{eq:defgabormtx} are assumed to be $\L$-periodic, therefore $x_\ell = x_{(\ell \mod \L)}$. 
Thus,~\eqref{eq:tfshiftssuperresradar} 
is of the form $\vy = \mA \vz$ with $\mA = \mG_{\vx} \herm{\mF}$. 
It is shown in~\cite{heckel_super-resolution_2015} that for a Gaussian probing signal $\vx$, the triplets $(b_k,\tau_k,\nu_k)$ can be recovered perfectly with high probability, provided the $[\tau_k,\nu_k]$ obey the minimum separation condition \eqref{eq:minsephdim} and the number of measurements satisfies $\L \geq c \S \log^3(\L)$. 
While the matrix $\mA = \mG_\vx \herm{\mF}$ 
does not obey the exact conditions in Section \ref{whatmeas}, 
the result in \cite{heckel_super-resolution_2015} can be derived from the mathematical framework developed in this paper (in particular from Lemma~\ref{prop:dualpolynomial} in Section \ref{sec:proofs}). 

We finally remark that previous works have considered the case where the $[\tau_k,\nu_k]$ lie on a coarse grid~\cite{herman_high-resolution_2009,heckel_identification_2013} (specifically a grid with spacing $(1/\L,1/\L)$).
Theorem 1 in \cite{krahmer_suprema_2014} guarantees that in this case, $\ell_1$-minimization provably succeeds.

\subsubsection{\label{sec:mimo}Super-resolution MIMO radar}

A standard or Single-Input Single-Output (SISO) radar can determine the relative distances and velocities of the targets by estimating the induced delays and Doppler shifts. 
A SISO radar is however incapable of determining the actual positions of the objects with a single transmit/receive antenna. 
MIMO radar systems \cite{bliss_multiple-input_2003,li_mimo_2007} 
use multiple antennas to transmit multiple probing signals simultaneously and record the reflections from the targets with multiple receive antennas. 
As a result a MIMO radar can resolve the relative angles along with the relative distances and velocities of targets. 
Let $\vx_j$ be the samples of the probing signal transmitted by the $j$-th transmit antenna, with $j = 0,\ldots,N_T-1$, where $N_T$ is the number of transmit antennas. 
As explained in more detail in \cite{heckel_super-resolution_2016b}, the response to those probing signals, recorded at the $r$-th antenna, $r=0,\ldots, N_R-1$, where $N_R$ is the number of receive antennas, is given by 
\begin{align}
\vy_r
&= 
\sum_{k=\indstart}^{\indend{\S}}
b_k' 
e^{i2\pi r N_T \beta_k } 
\sum_{j=0}^{N_T-1}
e^{i2\pi j \beta_k } 
\mc F_{\nu_k}
\mc T_{\tau_k}
\vx_j,
\label{eq:periorelMIMO}
\end{align}
where the parameters $\tau_k,\nu_k,\beta_k \in [0,1]$ correspond to delays, Doppler shifts, and angle. 
Again, by setting $\vy \defeq \transp{[ \transp{\vy}_{\rindstart},\ldots, \transp{\vy}_{\rindend{\R}} ] }$ and defining $\mAA \in \complexset^{N_R \L \times N_R N_T \L^2}$ appropriately, the input-output relation \eqref{eq:periorelMIMO} obeys \eqref{eq:syseqinwdim} with $\d=3$, as detailed in Section \ref{sec:proofmimo} (here, we assume for notational simplicity that $N_R N_T = \L$). In this paper, we prove that with high probability the parameters $(\beta_k,\tau_k,\nu_k)$ can be recovered perfectly, provided the minimum separation condition~\eqref{eq:minsephdim} holds, from a near-minimal number of measurements. 

\subsection{Outline}
In Section \ref{sec:recatomic} we present an algorithm for the generalized line spectral estimation problem, and Section \ref{sec:mainres} contains the corresponding formal recovery guarantees. 
In Section \ref{sec:discretesuperes} we prove results demonstrating that if the frequencies lie on an arbitrarily fine grid, $\ell_1$-minimization techniques provably succeed for generalized line spectral estimation. 
We illustrate the effectiveness of our methods with various numerical experiments in Section \ref{sec:numres}. Finally, we prove our results in Sections~\ref{sec:proofs} and~\ref{GordonProofs}.



\subsection{Notation}
We provide a brief summery of some of the notation used throughout the paper. 
The superscripts $\transp{}$ and $\herm{}$ stand for transposition and Hermitian transposition, respectively. 
For the vector $\vx$, $[\vx]_q$ denotes its $q$-th entry, $\twonorm{\vx}$ and $\onenorm{\vx}$ its $\ell_2$- and $\ell_1$-norm, and $\norm[\infty]{\vx} = \max_q |x_q|$ its largest entry.  
For the matrix $\mA$, $[\mA]_{ij}$ designates the entry in its $i$-th row and $j$-th column, 
 $\norm[]{\mA} \defeq\;$ $\max_{\norm[2]{\vv} = 1  } \norm[2]{\mA \vv}$ its spectral norm, and $\norm[F]{\mA} \defeq (\sum_{i,j} |[\mA]_{ij}|^2 )^{1/2}$ its Frobenius norm. 
The identity matrix is denoted by $\mI$. 
For convenience, we will frequently use a two- or three-dimensional index for vectors and matrices, e.g., we write $[\vg]_{(k,\ell)}, k,\ell=-\N,\ldots,\N$ for 
$
\vg = [g_{(-\N,-\N)}, g_{(-\N,-\N+1)}, \ldots,g_{(-\N,\N)},g_{(-\N+1,-\N)},\ldots,\allowbreak g_{(\N,\N)}]^T. 
$
For the set $\T$, $|\T|$ designates its cardinality and $\comp{\T}$ is its complement. 
Finally, $c,\tilde c, c', c_1,\allowbreak c_2,\ldots$ are numerical constants that can take on different values at different occurrences. 


\section{\label{sec:recatomic}Generalized line spectral estimation via atomic norm minimization}

In this section, we present the recovery algorithm considered in this paper. 
We consider the $\d$-dimensional case, and state formal results for $\d \leq 3$. 
Recall that recovery of the coefficients $\{b_k\}$ and frequencies $\{\vr_k\}$ from $\vz$ in~\eqref{eq:syseqinwdim} is a $\d$-dimensional line spectral estimation problem that can be solved
with standard spectral estimation techniques such as Prony's method \cite{stoica_spectral_2005}, \cite[Ch.~2]{gershman_space-time_2005}. Therefore,  it suffices to recover $\vz \in \complexset^{\L^\d}$ from the linear measurements $\vy=\mA\vz \in \complexset^{\R}$. 
To this aim, we use the fact that $\vz$ is a sparse linear combination of  atoms from the set $\setA \defeq \{ \vf(\vr), \vr \in [0,1]^\d\}$. 
A regularizer that promotes such a sparse linear combination is the atomic norm induced by these signals~\cite{chandrasekaran_convex_2012}, defined as 
\[
\norm[\setA]{\vz} 
\defeq \inf_{b_k \in \complexset, \vr_k \in [0,1]^\d} \left\{ \sum_k |b_k| \colon \vz = \sum_k b_k \vf(\vr_k) \right\}. 
\]
We estimate $\vz$ by solving the basis pursuit type atomic norm minimization problem   
\newcommand{\AN}{\mathrm{AN}}
\begin{align}
\AN(\vy) \colon \;\; \underset{\minlet{\vz}  }{\text{minimize}} \,  \norm[\setA]{\minlet{\vz} } \; \text{ subject to } \; \vy = \mAA \minlet{\vz}.
\label{eq:primal}
\end{align}
To summarize, we estimate $\{(b_k,\vr_k)\}$ from $\vy$ by 
\begin{enumerate}
\item solving $\AN(\vy)$ in order to obtain $\vz$, 
\item estimating $\{\vr_k\}$ from $\vz$ by solving the corresponding line spectral estimation problem, and 
\item solving the linear system of equations 
$
\vy = \sum_{k=\indstart}^{\indend{\S}} b_k \mAA \vf(\vr_k)
$ for the $b_k$. 
\end{enumerate}

We remark that the $\{\vr_k\}$ may be obtained more directly from a solution to the dual of \eqref{eq:primal};  
see \cite[Sec.~3.1]{bhaskar_atomic_2012}, \cite[Sec.~4]{candes_towards_2014}, \cite[Sec.~2.2]{tang_compressed_2013}, and \cite[Sec.~6]{heckel_super-resolution_2015} for details on this approach applied to related problems. 
This approach does, however, require identification of the zeros of a $\d$-dimensional trigonometric polynomial. 

Since computation of the atomic norm involves taking the infimum over infinitely many parameters, finding a solution to $\AN(\vy)$ may appear to be daunting. 
For the 1D case (i.e., $\d=1$), the atomic norm can be characterized in terms of linear matrix inequalities (LMIs) \cite[Prop.~2.1]{tang_compressed_2013}. This characterization is based on the Vandermonde decomposition lemma for Toeplitz matrices,  
and allows to formulate the atomic norm minimization program as a semidefinite program  that can be solved in polynomial time. While this lemma generalizes to higher dimensions \cite[Thm.~1]{yang_vandermonde_2015}, there is no polynomial bound on the dimension of the corresponding LMIs. Nevertheless, using~\cite[Thm.~1]{yang_vandermonde_2015} one can obtain a convex relaxation of $\AN(\vy)$, which is a Semidefinite Program (SDP) and can be solved in polynomial time. Similarly, a solution of the dual of $\AN(\vy)$ can be found with a SDP relaxation. 

Since the computational complexity of those SDPs is often  prohibitive in practice, in particular for $\d>1$, we do not detail these SDPs. Instead, we show in Section \ref{sec:discretesuperes} that the frequencies $\{\vr_k\}$ of the mixtures can be recovered on an arbitrarily fine grid via $\ell_1$-minimization. While this approach introduces a gridding error, the grid may be chosen sufficiently fine for the gridding error to be negligible compared to the error induced by additive noise (in practice, there is typically additive noise). 

\section{\label{sec:mainres} Main results}

In this section we state our main results which show that under mild conditions, the solution to $\AN(\vy)$ in \eqref{eq:primal} yields the unknown signal $\vz$. As discussed in the previous section, from $\vz$ we can then extract the coefficients and frequency parameters $\{(b_k,\vr_k)\}$. 
In the introduction, we informally stated our results for the 1D case. 
As mentioned in the introduction, those results continue to hold for higher dimensions. In an effort to minimize redundancy and for the sake of concreteness, we state and prove our results for isotropic and incoherent matrices for the 3D case. We choose the 3D case since the MIMO radar problem is a three dimensional problem, and  therefore the corresponding proof requires deriving technical results for the 3D case. 
We state results for three different classes of random matrices. 

\subsection{Isotropic and incoherent matrices}

We start by stating our main result for a matrix $\mA$ obeying isotropy and incoherence conditions. 
\begin{theorem}
Let $\mA \in \complexset^{\R \times \L^3}$ with $\L \geq 1024$, be a random matrix with rows $\va_r$ chosen independently from a distribution obeying the isotropy and incoherence properties
\begin{align}
\EX{\va_r \herm{\va}_r } = \frac{1}{\R} \mI\quad\text{and}\quad\sup_{\vf \in \complexset^{\L^3} \colon \infnorm{\vf} \leq 1}
\left|\innerprod{\vf}{\va_r}\right|^2
\leq \frac{\L^3}{\R} \mu,
\label{eq:incoherence}
\end{align}
for some fixed $\mu\geq 1$. Furthermore, assume we have $\R$ measurements $\vy = \mAA \vz$ from an unknown signal of the form $\vz = \sum_{k=\indstart}^{\indend{\S}} b_k \vf(\vr_k)$. Also assume that the signs $\sign(b_k)$ of the coefficients $b_k\in \complexset$ are chosen independently from symmetric distributions on the complex unit circle, and the vectors $\vr_k = (\beta_k,\tau_k,\nu_k) \in [0,1]^3$ obey the minimum separation condition
\begin{align}
\max( |\beta_k - \beta_{k'}|, |\tau_k - \tau_{k'}|, |\nu_k - \nu_{k'}| ) \geq \frac{5}{\N}
\text{ for all } k,k' \text{ with } k \neq k'.
\label{eq:mindistcond}
\end{align}
Then as long as
\begin{align}
\R \geq c\mu\S \log^2(\L/\delta),
\label{eq:coherenceass}
\end{align}
with $c$ a fixed numerical constant, $\vz$ is the unique minimizer of $\AN(\vy)$ in \eqref{eq:primal} with probability at least $1-\delta$. 
\label{thm:incoherence}
\end{theorem}
Our main results states that when the measurement matrix obeys the isotropy and incoherence conditions and the frequencies of the mixture are sufficiently separated then we can recover the signal $\vz$ via atomic norm minimization, and in turn recover the parameters $\{(b_k,\vr_k)\}$. 
We note that any algorithm (regardless of tractibility) requires at least $\R \gtrsim \S \log(\L/\S)$  measurements for \emph{stable} recovery, and $\R \geq \S$ measurements for recovery in general, so that our results are optimal up to the additional logarithmic factor. 

\subsection{Random projection matrices}
Our next class of matrices are random projection matrices that preserve Euclidean norms. These matrices are not covered directly by Theorem \ref{thm:incoherence} as they are not deterministically incoherent in the sense of definition \eqref{eq:incoherence}. However, these matrices are still incoherent in a probabilistic sense with respect to \emph{any} vector and in particular with respect to the vectors $\{\vf \in \complexset^\L\colon \infnorm{\vf} \leq 1 \}$. 

For simplicity we state our result for a Gaussian random matrix. 

\begin{theorem}
\label{thm:ISG}
Let $\mA \in \mc \complexset^{\R \times \L}$ be a random projection matrix with i.i.d.~$\mc N(0,1/\R)$ entries. 
Furthermore, assume we have $\R$ measurements $\vy = \mAA \vz$ from an unknown signal of the form $\vz = \sum_{k=\indstart}^\indend{\S} b_k \vf(\nu_k), \L \geq 512$. 
Also assume that the frequencies $\nu_k \in [0,1]$ obey the minimum separation condition~\eqref{eq:minsepintro1D}. 
Then as long as
\begin{align}
\R \geq c\S \log(\L),
\label{eq:condJLmtx} 
\end{align}
with $c$ a fixed numerical constant, $\vz$ is the unique minimizer of $\AN(\vy)$ in \eqref{eq:primal} with probability at least $e^{-(\R-2)/8}$.  
\end{theorem}

Note that this statement is superior to our previous statements in that it does not require the sign of the coefficients $\sign{(b_k)}$ to be random. 
The proof in Section \ref{GordonProofs} reveals that the result follows from Gordon's Lemma \cite{Gor,Gor2} carefully applied to this problem. Indeed, this result holds more broadly for all matrices that preserve Euclidean norms of vectors over a general set (i.e.,~matrices obeying Gordon's Lemma),
for instance sub-Gaussian random matrices and certain more structured random projection matrices defined below. 
The corresponding proof is analogous to that for a Gaussian matrix presented in Section \ref{GordonProofs} by using results in \cite{dirksen2015tail} and \cite[Thm.~3.3]{oymak2015isometric} ensuring that the corresponding random matrices 
behave like Gaussian random matrices in terms of dimensionality reduction. 

\paragraph{Isotropic sub-Gaussian matrices:}
The first is the class of isotropic sub-Gaussian matrices \cite{oymak2015sharp}. 
A matrix $\mA\in\mathbb{C}^{\R\times \L}$ is $\Delta$-sub-Gaussian if its rows are independent of each other and for all $r = \rindstart, \ldots, \rindend{\R}$, 
the $r$-th row $\herm{\ab}_r$ satisfies 
\[
\EX{\va_r}=\vect{0}, 
\quad \EX{ \va_r \herm{\va}_r } =\mI,
\footnote{The assumption that the population covariance matrix is identity is not necessary. In fact our result generalizes to any covariance matrix. The only change is that the required number of measurements will now scale with the condition number of the covariance matrix.}
\quad \text{ and for any vector $\vv$, $\PR{\left|\innerprod{\vv}{\va_r}\right| \geq t\twonorm{\vv}}\leq e^{-\frac{t^2}{\Delta^2}}$.}
\]
Theorem~\ref{thm:ISG} continues to hold with condition~\eqref{eq:condJLmtx} replaced by $\R \geq c\S \log(\L)/\Delta^2$.

\paragraph{Subsampled Orthogonal with Random Sign (SORS):}
The second class of matrices to which Theorem~\ref{thm:ISG} can be extend to is the class of SORS matrices~\cite{oymak2015isometric}. 
Let $\mU \in \reals^{\L\times \L}$ denote an orthonormal matrix obeying
\begin{align}
\label{BOS}
\herm{\mU} \mU = \mI 
\quad\text{and}\quad
\max_{i,j}\abs{\mtx{F}_{ij}}\le \frac{\Delta}{\sqrt{n}}.
\end{align}
Define the random subsampled matrix $\mH \in\reals^{\R\times \L}$ with i.i.d.~rows chosen uniformly at random from the rows of $\mU$. 
A SORS matrix is defined as $\mA=\mH \mD$, where $\mD \in \reals^{\L \times \L}$ is a random diagonal matrix with the diagonal entries i.i.d.~$\pm 1$ with equal probability. 
Theorem \ref{thm:ISG} extends to SORS matrices with condition~\eqref{eq:condJLmtx} replaced by
$
\R \geq c\frac{\S \log^4(\L)}{\Delta^2}
$. 
Note that such random matrices include the fast random projection matrices in~\cite{ailon_almost_2013,krahmer_new_2011}. 
Specifically, with $\mH$ the Hadamard matrix, the random projection $\mA  = \mH \mD$ can be applied in time $O( \L \log \L)$ as opposed to time $O(\R \L)$ for the realizations of general sub-Gaussian random matrices.



\subsection{Superresolution MIMO radar matrices}
We next state our main result for random matrices appearing in MIMO radar problems. 
A simple calculation shows (see Section \ref{sec:proofmimo}) that the MIMO radar input-output relation \eqref{eq:periorelMIMO} can be viewed as taking $N_R \L$ linear measurements 
\begin{align}
\label{myeq}
\vy=\begin{bmatrix}\vy_\rindstart \\\vy_\rindsecond \\\vdots\\\vy_{\rindend{N_R}}\end{bmatrix}
=\mA\vz,
\end{align}
from a vector $\vz$ of the form \eqref{eq:syseqinwdim} with $\d=3$. 
Here, $\mAA \in \complexset^{N_R \L \times N_R N_T \L^2}$ is a function of the transmitted signals $\vx_j$ and is given by
\begin{align}
\label{eq:defAmimo}
\mA \defeq \begin{bmatrix}
\mB &  &  \\
& \ddots & \\
 & & \mB 
\end{bmatrix},
\quad
\mB \defeq [\mG_{\vx_{0}} \herm{\mF} \ldots  \mG_{\vx_{N_T - 1 } } \herm{\mF} ] 
\in 
\complexset^{\L \times \L^2 N_T},
\end{align}
where $\mG_{\vx} \herm{\mF}$ is the matrix in \eqref{eq:tfshiftGF} that appeared in the super-resolution radar problem before. 
\begin{theorem} 
\label{thm:mainresmimo}
Assume that the probing signals $\vx_0, \vx_\indsecond,\ldots,\vx_{N_T-1} \in \reals^{\L}$ are i.i.d. $\mathcal{N}(\vct{0},\frac{1}{N_T\L}\mtx{I})$ random vectors with $\L = 2\N+1 \geq 1024$ and $N_T N_R \geq 1024$. Let $\vy = \mAA \vz$ with $\mA$ defined in~\eqref{eq:defAmimo} and $\vz = \sum_{k=\indstart}^{\indend{\S}} b_k \vf(\vr_k)
$. 
Also, assume that the signs $\sign(b_k)$ of the coefficients $b_k\in \complexset$ are chosen independently from symmetric distributions on the complex unit circle, and the vectors $\vr_k = (\beta_k,\tau_k,\nu_k) \in [0,1]^3$ obey the minimum separation condition
\begin{align}
|\beta_k - \beta_{k'}| \geq \frac{10}{N_T N_R- 1} \quad \text{or}\quad
 |\tau_k - \tau_{k'}| \geq \frac{5}{\N} \quad\text{or}\quad
|\nu_k - \nu_{k'}|  \geq \frac{5}{\N}
, \quad \text{for all $k,k'\colon k\neq k'$}. 
\label{eq:minsepcond}
\end{align}
Then as long as
\begin{align}
\label{eq:sleqb}
\min(\L, N_T N_R) \geq \S c \log^3\left( \L / \delta \right),
\end{align}
with $c$ a fixed numerical constant, $\vz$ is the unique minimizer of $\AN(\vy)$ in \eqref{eq:primal} with probability at least $1-\delta$. 
\end{theorem}
Even though this result does not allow $\S$ to be near linear in the total number of measurements $N_R \L$, it is optimal up to log-factors. Indeed, $\S \leq \min(\L, N_TN_R)$ is in general a necessary condition to uniquely recover the $b_k$ even when the frequencies of the mixture components $(\beta_k,\tau_k,\nu_k)$ are known \cite{heckel_super-resolution_2016b}. 

We note that Theorems \ref{thm:incoherence} and \ref{thm:mainresmimo} assume that the coefficients $\{b_k\}$ have random sign. 
In the radar model, where the $b_k$ describe the attenuation factors, this is a common and well justified assumption~\cite{bello_characterization_1963}. 
While this assumption is reasonable for several applications, we believe that it is not necessary for our results to hold. Indeed, Theorem~\ref{thm:ISG} does not require this assumption. We hope to remove this assumption in future publications.
\section{\label{sec:discretesuperes} Recovery on a fine grid}
A natural approach for estimating the parameters $\{\vr_k\}$ from $\vy = \mA \vz$, with $\vz$ defined in~\eqref{eq:syseqinwdim}, is to assume the parameters $\{\vr_k\}$ lie on a \emph{fine} grid, and solve the generalized line spectral estimation problem on that grid. When the frequencies do not lie on a grid this of course leads to a ``gridding error". However, this error is negligible when the grid is very fine. When the frequencies do lie on a grid the generalized line spectral estimation problem can be be formulated as a discrete sparse signal recovery problem, albeit with a very coherent dictionary. 
In this section we discuss this approach in detail focusing on the 3D-case. 

Suppose the vectors $\vr_k = (\beta_k,\tau_k,\nu_k)$ lie on a grid with spacing $(1/K_1,1/K_2,1/K_3)$, where $K_1,K_2,K_3$ are integers obeying $K_1,K_2,K_3 \geq \L = 2 \N +1$, 
see Figure \ref{fig:illgrid} for an illustration. 
With this assumption, 
recovery of the parameters $\{(b_k,\vr_k)\}$ from $\vy$ reduces to the recovery of the sparse (discrete) signal $\vs\in \complexset^{K_1 K_2 K_3}$ 
from the measurement
\[
\vy = \mA \mF_{\mathrm{grid}} \vs, 
\]
where $\mF_{\mathrm{grid}} \in \complexset^{\L^3  \times K_1 K_2 K_3}$ is the matrix with $(n_1,n_2,n_3)$-th column given by $ \vf(n_1/\K_1, n_2/\K_2, n_3/\K_3)$. 
The non-zeros of $\vs$ and its indices correspond to the $\{b_k\}$ and the $\{\vr_k\}$ on the grid, respectively. 
A standard approach for recovering the sparse signal $\vs$ from the underdetermined linear system of equations $\vy = \mA \mF_{\mathrm{grid}} \vs$ is to solve the following convex program:
\begin{align}
	\mathrm{L1}(\vy) \colon \;\;
\underset{\minlet{\vs}}{\text{minimize}} \; \onenorm{\minlet{\vs}} \text{ subject to } \vy = \mA \mF_{\mathrm{grid}} \minlet{\vs}.
\label{eq:l1minmG}
\end{align}

\begin{figure}
\begin{center}
	\begin{tikzpicture}[scale=1.8,>=latex] 
	\draw[->] (0,3) -- (2.3,3); 
	\draw[->] (0,3) -- (0,5.3);
	\node at (1.2,2.9) [anchor=north] {$\tau$};
	\node at (-0.1,4.2) [anchor=south, rotate=90] {$\nu$};	
	
	\foreach \x in {0.05,0.15,...,2} 
	\foreach \y in {3.05,3.25,...,5} 
	\fill [black,opacity=1] (\x,\y) circle (0.013);

	\fill [color=DarkBlue,] (1.05,3.65) circle (0.04);
	\fill [color=DarkBlue,] (0.65,4.05) circle (0.04);
	\fill [color=DarkBlue,] (1.35,3.85) circle (0.04);

	\coordinate (a) at (3,4);
	\coordinate (b) at (4,4);
	\coordinate (c) at (4,5);
	\coordinate (d) at (3,5);

	\draw [densely dashed,opacity=0.6] (2,4.25) -- (2.65,4.25);
	\draw [densely dashed,opacity=0.6] (2,4.85) -- (2.65,4.85);
	\draw [<->,opacity=0.6,>=latex]  (2.65,4.25) -- (2.65,4.85);
	\node at (2.65,4.55) [anchor=west] {$\frac{1}{\L}$};
	
	\draw [->,opacity=0.6,>=latex]  (2.2,5.1) -- (2.2,4.85);
	\draw [<-,opacity=0.6,>=latex]  (2.2,4.65) -- (2.2,4.4);
	\draw [densely dashed,opacity=0.6] (2,4.65) -- (2.2,4.65);
	\node at (2.4,4.48) [anchor=south] {$\frac{1}{K_2}$};
	
		
	\draw [densely dashed,opacity=0.6] (1.15,5.45) -- (1.15,4.85);
	\draw [densely dashed,opacity=0.6] (1.25,5) -- (1.25,4.85);
	\draw [->,opacity=0.6,>=latex]  (0.9,5) -- (1.15,5);
	\draw [<-,opacity=0.6,>=latex]  (1.25,5) -- (1.5,5);
	\node at (1.3,5) [anchor=south] {$\frac{1}{K_1}$};
	
	\draw [densely dashed,opacity=0.6] (1.75,5.45) -- (1.75,4.85);
	\draw [<->,opacity=0.6,>=latex]  (1.15,5.45) -- (1.75,5.45);
	\node at (1.25,5.65) [anchor=west] {$\frac{1}{\L}$};
	
\end{tikzpicture}
\end{center}
\caption{\label{fig:illgrid}Illustration of the (fine) grid in two dimensions. In this example, the non-zeros, marked by blue dots, violate the minimum separation condition \eqref{eq:discreteminsep}.}
\end{figure}
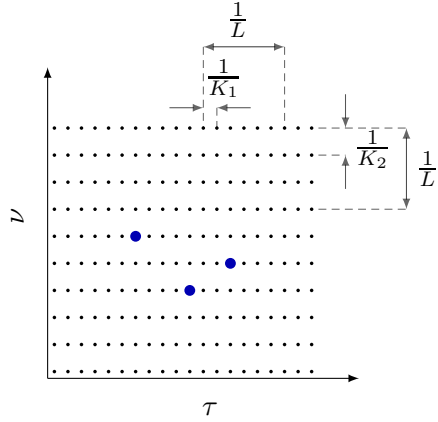
\noindent The following theorems discuss our main results for recovery on the fine grid. For brevity we only state the discrete analagoue to Theorems \ref{thm:incoherence} and \ref{thm:mainresmimo}.
A similar result holds for Theorem~\ref{thm:ISG}. Our first result is the discrete analogue of Theorem \ref{thm:incoherence}.
\begin{theorem} 
\label{cor:discretesuperresII}
Let $\vs$ be an $\S$-sparse vector with non-zeros indexed by the support set $\mc S \subseteq [K_1]\times [K_2] \times [K_3], [K] \defeq \{0,\ldots K-1\}$, 
obeying the minimum separation condition
\begin{align}
\max_{\ell \in\{1,2,3\}}\frac{|n_{\ell} - n_{\ell}'|}{K_\ell}  \geq \frac{5}{\N},
\quad 
\text{for all triplets }
[n_1,n_2,n_3], [n_1',n_2',n_3'] \in \mc S. 
\label{eq:discreteminsep}
\end{align}
Furthermore, assume that the signs of the non-zeros of $\vs$ are chosen independently from symmetric distributions on the complex unit circle. Also assume we have $\R$ linear measurements $\mA \in \complexset^{\R \times \L^3}$ with $\L \geq 1024$ of the form $\vy = \mA \mF_{\mathrm{grid}} \vs$ with the matrix $\mA$ obeying the conditions of Theorem \ref{thm:incoherence}, in particular the isotropy and incoherence conditions. Then as long as 
\begin{align*}
\R \geq c \mu\S\log^2(\L/\delta),
\end{align*}
with $c$ a fixed numerical constant, $\mathbf{s}$ is the unique minimizer of $\mathrm{L1}(\vy)$ with probability at least $1-\delta$. 
\end{theorem}

Our second result is the discrete analogue of Theorem \ref{thm:mainresmimo}.
\begin{theorem} 
Let $\vs$ be a $\S$-sparse vector with non-zeros indexed by the support set $\mc S \subseteq [K_1]\times [K_2] \times [K_3]$, 
obeying the minimum separation condition, for all triplets $[n_1,n_2,n_3], [n_1',n_2',n_3'] \in \mc S$, 
\[
\frac{|n_{1} - n_{1}'|}{K_1} \geq \frac{10}{N_T N_R- 1} 
\quad \text{or}\quad
\frac{|n_{2} - n_{2}'|}{K_1} \geq \frac{5}{\N}
\quad\text{or}\quad
\frac{|n_{3} - n_{3}'|}{K_1}  \geq \frac{5}{\N}. 
\]
Furthermore, assume that the signs of the non-zeros of $\vs$ are chosen independently from symmetric distributions on the complex unit circle. Suppose that $\mA$ satisfies the conditions of Theorem \ref{thm:mainresmimo} and assume \eqref{eq:sleqb} holds for some fixed $\delta > 0$. 
Then, with probability at least $1-\delta$, $\mathbf{s}$ is the unique minimizer of $\mathrm{L1}(\vy)$. 
\label{cor:mimodiscrete}
\end{theorem}
Note that Theorems \ref{cor:discretesuperresII} and \ref{cor:mimodiscrete} do not impose any restrictions on $K_1,K_2,$ and $K_3$, in particular they can be arbitrarily large.  
The proof of those results, presented in Appendix \ref{sec:proofdiscrete}, is closely linked to that of Theorems \ref{thm:incoherence} and \ref{thm:mainresmimo}. Specifically, the existence of a certain dual certificate guarantees that $\vs$ is the unique minimizer of $\mathrm{L1}(\vy)$ in \eqref{eq:l1minmG}. 
These dual certificates are obtained directly from the dual polynomials for the continuos case in the proofs of Theorems \ref{thm:incoherence} and \ref{thm:mainresmimo}. See Appendix \ref{sec:proofdiscrete} for further details.
\section{\label{sec:numres} Numerical results}
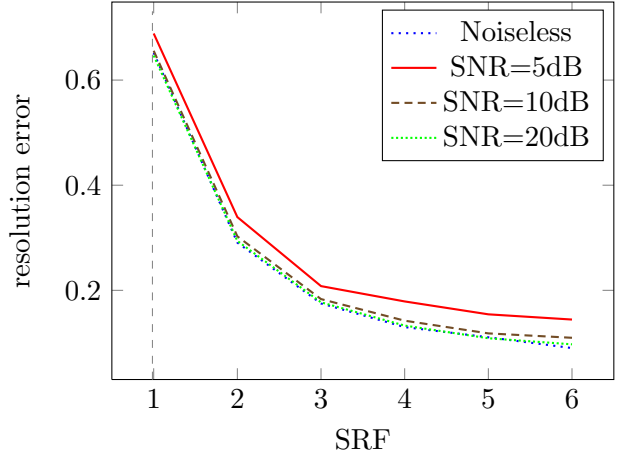
\begin{figure}[t]
\begin{center}
 \begin{tikzpicture}[scale=1,thick] 
 \begin{axis}[xlabel={$\text{SRF}$}, ylabel = {resolution error},  xtick={1,2,3,4,5,6},
 width = 0.5\textwidth, height=0.4\textwidth]   
 \addplot +[thick,mark=none, dotted] table[x index=0,y index=1]{./fig/offgrid.dat}; 
 \addlegendentry{Noiseless}
 \addplot +[thick,mark=none] table[x index=0,y index=2]{./fig/offgrid.dat};
 \addlegendentry{SNR=5dB}
  \addplot +[thick,mark=none,densely dashed] table[x index=0,y index=3]{./fig/offgrid.dat}; 
 \addlegendentry{SNR=10dB}
  \addplot +[thick,mark=none, green, densely dotted] table[x index=0,y index=4]{./fig/offgrid.dat};
 \addlegendentry{SNR=20dB} 
  \end{axis}
  \draw[very thin,dashed,gray] (0.54,0) -- (0.54,4.9);
 \end{tikzpicture}
 \end{center}
 \vspace{-0.4cm}
 \caption{\label{fig:realrecov} 
 Resolution error for the recovery of $\S = 5$ targets from the samples $\mathbf y$ with and without additive Gaussian noise $\mathbf n$ of a certain signal-to-noise ratio $\text{SNR} = \| \mathbf y \|^2_2/ \| \mathbf n \|^2_2$, for varying super-resolution factors (SRFs).  
 The resolution error is defined as the average over $( N_T^2 N_R^2(\hat \beta_k - \beta_k)^2 + \L^2(\hat \tau_k - \tau_k)^2 + \L^2(\hat \nu_k - \nu_k)^2)^{1/2}$, $k=\indstart,\ldots,\indend{\S}$, where $(\hat \beta_k,\hat \tau_k, \hat \nu_k)$ are the locations obtained by solving \text{L1-ERR}.
}
\end{figure}
In this section, we perform some numerical experiments to evaluate the resolution obtained by our approach, to demonstrate that it is robust to noise, and also to compare it to a competing algorithm. We focus on the MIMO radar problem, as results for the other matrices covered by our theory are similar.  
We set $N_T = 3, N_R = 3, \L = 41$, and draw $S=5$ target locations $(\beta_k,\tau_k,\nu_k)$ 
and corresponding attenuation factors $b_k$
uniformly at random from $[0,1] \times [0,2/\sqrt{\L}]^2$
and from the complex unit disc, respectively. 
Moreover, we choose $K_1 = \SRF \cdot N_T N_R, K_2 = \SRF \cdot \L$, and $K_3 = \SRF \cdot \L$, where $\SRF\geq 1$ can be interpreted as a super-resolution factor determining by how much the $(1/K_1,1/K_2,1/K_3)$ grid is finer than the ``coarse" grid $(1/(N_TN_R),1/\L,1/\L)$. 
To account for additive noise, we solve the following modification of $\mathrm{L1}(\vy)$ in \eqref{eq:l1minmG}
\begin{align}
\text{L1-ERR}\colon \underset{\minlet{\vs}}{\text{minimize}} \; \onenorm{\minlet{\vs}} \text{ subject to } 
\twonorm{\vy - \mR \minlet{\vs} }^2 \leq \eta,
\label{eq:BDDN}
\end{align}
with $\eta$ chosen on the order of the noise variance. 
There are two error sources incurred by this approach: the gridding error obtained by assuming the points lie on a grid with spacing\linebreak $(1/K_1,1/K_2,1/K_3)$, and the additive noise error. 
The former decreases when $\SRF$ increases, while the latter only depends on the noise variance. 
The results, depicted in Figure \ref{fig:realrecov}, show that the target resolution of our approach is significantly better than the compressed sensing based approach proposed in \cite{dorsch_refined_2015,strohmer_adventures_2015} corresponding to recovery on the coarse grid, i.e.,~$\SRF=1$. Moreover, the results show that our approach is robust to noise. 

We now compare our approach to the  Iterative Adaptive Approach (IAA) \cite{yardibi_source_2010}, proposed for MIMO radar in~\cite{roberts_iterative_2010}. 
IAA is based on weighted least squares and has been proposed in the array processing literature, where typically multiple measurements, i.e., multiple observations $\vy$ are available. 
IAA  can work well with even one measurement vector $\vy$ and can therefore be directly applied to the problems considered in this paper. 
However, to the best of our knowledge no rigorous performance guarantees are available for IAA. 
We compare the IAA algorithm \cite[Table II, ``The IAA-APES Algorithm'']{yardibi_source_2010} to L1-ERR in \eqref{eq:BDDN}. We use the same problem parameters as above, but with $\SRF = 3$ and $(\beta_k,\tau_k,\nu_k) = (k/(N_R N_t), k/\L, k/\L), k=0,\ldots,\S-1$, so that the location parameters lie on the coarse grid and are therefore separated. 
As before, we draw the corresponding attenuation factors $b_k$ i.i.d.~uniformly at random from the complex unit disc. 
Our results, depicted in Figure \ref{fig:cmpmusicconvex}, show that L1-ERR performs better in this experiment than IAA specially for small signal-to-noise ratios.

\begin{figure}[!ht]
\centering
\begin{tikzpicture}    

\begin{axis}[
xlabel = SNR in dB,
ylabel=resolution error,
x dir=reverse,
legend pos=outer north east,
yticklabel style={/pgf/number format/fixed,
                  /pgf/number format/precision=3},
                  width = 0.45\textwidth,
]

%

\addplot +[blue] table[x index=0,y index=1]{./fig/resultsSRF_nfp_26aug.dat};
\addlegendentry{ L1-ERR};

\addplot +[red] table[x index=0,y index=2]{./fig/resultsSRF_nfp_26aug.dat};
\addlegendentry{IAA};

\end{axis}

\end{tikzpicture}

\caption{
\label{fig:cmpmusicconvex}
Resolution error achieved by L1-ERR and by IAA applied to $\vy + \vn$, where $\vn \in \complexset^R$ is additive Gaussian noise, such that the signal-to-noise ratio is 
 $\text{SNR} \defeq \twonorm{\vy}^2 / \twonorm{ \vn }^2$. 
 As before, the resolution error is defined as $( N_T^2 N_R^2(\hat \beta_k - \beta_k)^2 + \L^2(\hat \tau_k - \tau_k)^2 + \L^2(\hat \nu_k - \nu_k)^2)^{1/2}$, 
where $(\hat \beta_k,\hat \tau_k, \hat \nu_k)$ are the locations obtained by using \text{L1-ERR} or IAA. 
The resolution error is averaged over $100$ trials. 
}

\end{figure}
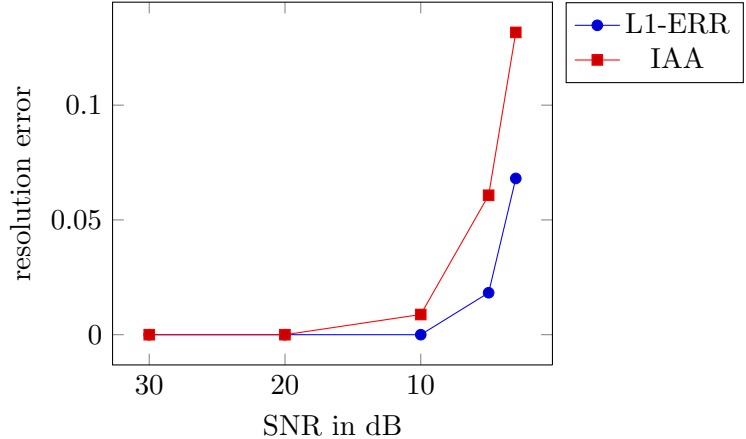




\section{\label{sec:proofs}Proofs of Theorems \ref{thm:incoherence} and \ref{thm:mainresmimo}}

Theorems \ref{thm:incoherence} and \ref{thm:mainresmimo} are proven by constructing an appropriate dual certificate. The existence of this  certificate guarantees that the solution to $\AN(\vy)$ in \eqref{eq:primal} is $\vz$, as formalized by 
Proposition \ref{prop:dualmin} below. 
Proposition \ref{prop:dualmin} is a consequence of strong duality, and well known for the discrete setting from the compressed sensing literature \cite{candes_robust_2006}. 
The proof of Proposition \ref{prop:dualmin} is standard, see for example~\cite[Proof of Prop.~2.4]{tang_compressed_2013}. 
\begin{proposition} 
Let $\vy=\mAA \vz$ with $\vz =   \sum_{k=\indstart}^{\indend{\S}}  b_k  \vf(\vr_k)$. If there exists a dual certificate
\begin{align}
Q(\vr) =  
 \innerprod{ \herm{\mAA}\vq}{\vf(\vr) } 
\label{eq:dualpolyinprop}
\end{align}
with complex coefficients $\vq \in \complexset^{\R}$ such that 
\begin{align}
Q(\vr_k) = \sign(b_k), \text{ for all } k = \indstart,\ldots,\indend{\S}, 
\quad
\text{and } |Q(\vr)| < 1 \text{ for all } \vr \in [0,1]^3 \setminus \{\vr_{\indstart},\ldots,\vr_{\indend{\S}}\}, 
\label{eq:dualpolyinatmincon}
\end{align}
then $\vz$ is the unique minimizer of $\AN(\vy)$. 
\label{prop:dualmin}
\end{proposition}
\noindent We now need to show that such a dual certificate exists under appropriate conditions. This is the subject of the next section.

\subsection{\label{sec:statementlem1}Sufficient conditions for the existence of dual certificates}

In this section we state the main technical result of this paper---Lemma~\ref{prop:dualpolynomial} below---which provides sufficient conditions on $\mA$ that guarantee the existence of a dual certificate as described in Proposition~\ref{prop:dualmin}. 
The proofs of Theorems \ref{thm:incoherence} and \ref{thm:mainresmimo} are based on verifying these conditions. 

To derive this lemma, we construct the dual certificate $Q$ of Proposition \ref{prop:dualmin} explicitly. 
Our construction of $Q$ is similar to that in \cite{heckel_super-resolution_2015}, 
and is inspired by the construction of related polynomials in \cite{candes_towards_2014,tang_compressed_2013}. 
We briefly outline the construction of $Q$, in order to introduce the notation required for stating our key lemma. 
To this aim, first note that the constraint \eqref{eq:dualpolyinprop} requires $Q$ to be a 3D-trigonometric polynomial with coefficients $\herm{\mA}\vq$. To build $Q$ we therefore need to construct a 3D-trigonometric polynomial that satisfies the conditions~\eqref{eq:dualpolyinatmincon}, and whose coefficients are constraint to be of the form $\herm{\mAA}\vq$, for some $\vq \in \complexset^{\R}$. 
Since $\mA$ is a random matrix, $Q$ is a \emph{random} trigonometric polynomial.   

It is instructive to first consider the construction of a \emph{deterministic} 3D trigonometric polynomial of the form $\bar Q(\vr) = \innerprod{\bar \vq}{\vf(\vr)}$ with deterministic coefficients $\bar \vq$ 
that satisfies the conditions~\eqref{eq:dualpolyinatmincon}, but whose coefficients $\bar \vq \in \complexset^{\L^3}$ are \emph{not} required to be of the form $\herm{\mA} \vq$. 
Such a construction has been established (provided a minimum separation condition on the $\vr_k$ holds) by \cite[Prop.~2.1, Prop.~C.1]{candes_towards_2014} for the 1D and 2D case. The 3D result follows in a similar manner, we therefore omit the details. 
To construct $Q(\vr)$, Cand\`es and Fernandez-Granda \cite{candes_towards_2014} interpolate the points $u_k \defeq \sign(b_k)$ with a fast-decaying kernel $\bar G$ (defined below) and slightly adopt this interpolation near the $\vr_k$ with the partial derivatives $\bar G^{(n_1,n_2,n_3)}(\vr) \defeq 
\frac{\derd^{n_1} }{ \derd \beta^{n_1}}
\frac{\derd^{n_2} }{ \derd \tau^{n_2}} \frac{\derd^{n_3} }{ \derd \nu^{n_3}}  
\bar G(\vr)$ (recall that $\vr = \transp{\vectbr{\beta,\tau,\nu}}$) to ensure that local maxima are achieved at the $\vr_k$. Specifically, $\bar{Q}$ takes the form 
\begin{align}
\bar Q(\vr) = \sum_{k=\indstart}^{\indend{\S}} \bar \alpha_k \bar G(\vr-\vr_k) + \bar \alpha_{1k} \bar G^{(1,0,0)}(\vr - \vr_k) +  \bar\alpha_{2k} \bar G^{(0,1,0)}(\vr - \vr_k)
+  \bar\alpha_{3k} \bar G^{(0,0,1)}(\vr - \vr_k). 
\label{eq:detintpolC}
\end{align}
Here, $\bar G^{(n_1,n_2,n_3)}(\vr) \defeq \frac{\derd^{n_1} }{ \derd \beta^{n_1}} \frac{\derd^{n_2} }{ \derd \tau^{n_2}} \frac{\derd^{n_3} }{ \derd \nu^{n_3}} \bar G(\vr)$ and 
$  
\bar G(\vr)
\defeq \FK(\tau) \FK(\nu) \FK(\beta)
$,
where $\FK$ is the squared Fej\'er kernel defined as
\[
 \FK(t) \defeq \left( \frac{\sin\left( \M \pi t\right)}{\M \sin(\pi t)} \right)^4, \quad \M \defeq \frac{\N}{2}+1.
\]
The Fej\'er kernel is a trigonometric polynomial of degree $\N/2$ which implies that $\FK$ is a trigonometric polynomial of degree $\N$. $\FK$ can be written in the form  
\begin{align}
\label{eq:def:FejK}
 \FK(t) = \frac{1}{\M} \sum_{k=-\N}^{\N}  g_k e^{i2\pi t k},
\end{align}
with coefficients $g_k$ obeying $|g_k| \leq 1$. 
Observe that shifted versions of $\FK$ (i.e.,~$\FK(t - t_0)$) and the derivatives of $\FK$ are also 1D trigonometric polynomials of degree $\N$. Therefore the kernel $\bar G$, its partial derivatives, and shifted versions are also 3D trigonometric polynomials of the form $\innerprod{\bar \vq }{\vf(\vr)}$, as desired. 

$\bar{Q}$ is not a valid dual certificate for the generalized line spectral estimation problem as it is not of the form \eqref{eq:dualpolyinprop}. However, inspired by the construction of $\bar{Q}$ we build $Q$ by interpolating the points $u_k = \sign(b_k)$ at $\vr_k$ with the functions 
\begin{align}
G_{\vm}(\vr,\vr_k) 
\defeq 
\innerprod{\mA \vg_{\vm}(\vr_k) }{ \mA \vf(\vr) }. 
\label{eq:Gstructuredef}
\end{align}
Here, $\vg_{\vm}(\vr), \vm \defeq \transp{[m_1,m_2,m_3]}$ is the vector with $(v,k,p)$-th coefficient given by
\begin{align}
\frac{1}{\M^3}
g_v g_k g_p
(i2\pi v)^{m_1}(i2\pi k)^{m_2} (i2\pi p)^{m_3}
e^{-i2\pi (\beta v + \tau k + \nu p)},  
\label{eq:defvge}
\end{align}
where the $g_k$ are the coefficients of the squared Fej\'er kernel $\FK$ in~\eqref{eq:def:FejK}. 
Assuming $\EX{\herm{\mA} \mA } =\mI$ we have 
\begin{align}
\EX{G_{\vm}(\vr, \vr_k)} 
=
\frac{\derd^{m_1}}{ \derd \beta^{m_1}} 
\frac{\derd^{m_2}}{ \derd \tau^{m_2}} 
\frac{\derd^{m_3}}{ \derd \nu^{m_3}} 
\FK(\beta-\beta_k) 
\FK(\nu-\nu_k) 
\FK(\tau-\tau_k)
=  
\bar G^{(\vm)}(\vr-\vr_k). 
\label{eq:exgmn}
\end{align}
We construct $Q$ by interpolating the $u_k$ at $\vr_k$ with the functions $G_{(0,0,0)}(\vr,\vr_k), k=\indstart,\ldots,\indend{\S}$, 
and slightly adopt this interpolation near the $\vr_k$ with linear combinations of the functions $G_{(1,0,0)}(\vr,\vr_k)$, $G_{(0,1,0)}(\vr,\vr_k)$, and $G_{(0,0,1)}(\vr,\vr_k)$, in order to ensure that local maxima of $Q$ are achieved exactly at the $\vr_k$: 
\begin{align}
Q(\vr) = \sum_{k=\indstart}^{\indend{\S}} 
\alpha_k G_{(0,0,0)}(\vr,\vr_k)
+ \alpha_{1k} G_{(1,0,0)}(\vr,\vr_k)
+ \alpha_{2k} G_{(0,1,0)}(\vr,\vr_k)
+ \alpha_{3k} G_{(0,0,1)}(\vr,\vr_k).
\label{eq:dualpolyorig}
\end{align}  
Since $Q(\vr)$ is a linear combination of the functions $G_{\vm}(\vr,\vr_k) = \innerprod{\herm{\mA}\mA \vg_{\vm}(\vr_k) }{ \vf(\vr) }$, 
it is of the form $\innerprod{\herm{\mA}\vq}{\vf(\vr)}$ as required by \eqref{eq:dualpolyinprop}. 

We are now ready to state our key lemma, whose proof is deferred to Appendix~\ref{sec:prooflemdualpoly}, providing sufficient conditions on $\mA$ ensuring that $Q$ is indeed an appropriate dual certificate. 
\begin{lemma}
\label{prop:dualpolynomial}
Let $\T = \{\vr_{\indstart}, \vr_1,\ldots,\vr_{\indend{\S}}\} \subset [0,1]^3$ be an arbitrary set of points obeying the minimum separation condition \eqref{eq:mindistcond}. 
Let $\vu \in \complexset^\S$ be a random vector, whose entries are chosen independently from symmetric distributions on the complex unit circle.  
Pick $\delta>0$, and let $\mA$ be a random matrix that obeys
\begin{subequations}
\begin{align}
\PR{
\frac{1}{\kappa^{\onenorm{\vm + \vn}}} 
\abs{
\innerprod{ (\herm{\mA}\mA - \mI) \vg_\vm(\vr_k)}{ \vf^{\vn}(\vr)}
}
\geq \frac{c_1}{\sqrt{S \log ( \L / \delta )   }   }
}
\leq  c_3 \frac{ \delta}{\L^{15}},
\label{eq:cond1new}
\end{align}
for all $\vr_k, \vr \in [0,1]^3$, all $\vm,\vn \in \{0,1,2\}^3$ with $\onenorm{\vm},\onenorm{\vn}\leq 2$. \\
Here, $[\vf^{\vn}(\vr)]_{(v,k,p)} \defeq  (-i2\pi v)^{n_1} (-i2\pi k)^{n_2} (-i2\pi p)^{n_3}  e^{-i2\pi (\beta v + \tau k + \nu p)}$ and $\kappa \defeq \sqrt{|\FK^{(2)}(0)|}$. 
Moreover, suppose that for $\vr_k \in [0,1]^3$, there is a constant $\hat c$ such that 
\begin{align}
\PR{
\max_{\vr \in [0,1]^3}  \left | \innerprod{\mA \vg_{\vect{0}}(\vr_k)}{\mA \vf(\vr)} \right| \geq  \frac{\hat c}{S} \L^{3}  
} \leq \frac{\delta}{8}. 
\label{eq:GrmGrgA}
\end{align}
\end{subequations}
Then, with probability at least $1-\delta$, there exists a trigonometric polynomial $Q(\vr) =  \innerprod{ \herm{\mAA}\vq}{\vf(\vr) }, \vr = \transp{\vectbr{\tau,\nu,\beta}}$, with complex coefficients $\vq \in \complexset^{\R}$ obeying
\begin{align}
Q(\vr_k) =  u_k, \text{ for all } \vr_k \in \T, \text{ and } |Q(\vr)| < 1 \text{ for all } \vr \in [0,1]^3 \setminus \T.
\label{eq:polyboundedinprop}
\end{align}
\end{lemma}

Theorems \ref{thm:incoherence} and \ref{thm:mainresmimo} 
are proven in the coming sections by showing that the conditions \eqref{eq:cond1new} and \eqref{eq:GrmGrgA} are satisfied for the corresponding random matrices $\mA$.



\subsection{\label{sec:proofthm1}Proof of Theorem~\ref{thm:incoherence}: Isotropic and incoherent matrices}
We prove Theorem \ref{thm:incoherence} by showing that the isotropy and incoherence conditions imply conditions~\eqref{eq:cond1new} and \eqref{eq:GrmGrgA} of 
Lemma~\ref{prop:dualpolynomial}, which in turn implies the existence of an appropriate dual certificate satisfying the conditions of Proposition \ref{prop:dualmin}. 

We start with verifying \eqref{eq:cond1new}. Note that 
\begin{align}
\frac{1}{\kappa^{\onenorm{\vm + \vn}}}\innerprod{ (\herm{\mA} \mA - \mI) \vg_\vm(\vr_k) }{\vf^{\vn}(\vr)}
=
\sum_{r=\rindstart}^{\rindend{\R}} 
\underbrace{ 
\frac{1}{\kappa^{\onenorm{\vm + \vn}}}
\herm{(\vf^{\vn}(\vr))}\left( \va_r \herm{\va}_r  - \frac{1}{\R} \mI \right) \vg_\vm(\vr_k) }_{v_r},
\label{eq:sumrar}
\end{align}
where $\herm{\va}_r$ are the rows of $\mA$. 
From the isotropy property ($\EX{\va_r \herm{\va}_r } = \frac{1}{\R} \mI$), we have $\EX{v_r} = 0$. 
Thus, \eqref{eq:sumrar} is a sum of independent zero-mean random variables, and we will use Bernstein's inequality to provide an upper bound. Define $\tilde \vf \defeq \frac{1}{\kappa^{\onenorm{\vn}} \sqrt{\L^3} } \vf^{\vn}(\vr)$ and $\tilde \vg \defeq \sqrt{\L^3} \frac{1}{\kappa^{\onenorm{\vm}}} \vg_\vm(\vr_k)$, for notational convenience. 
In order to apply Bernstein's inequality we need an upper bound on $|v_r|$. 
To derive this bound, first note that a simple application of H\"older's inequality yields
\begin{align}
|v_r|
&\leq 
\left| \herm{\tilde \vf} \va_r \herm{\va}_r  \tilde \vg \right| \leq 
\infnorm{\tilde \vf} \infnorm{\tilde \vg} \onenorm{\va_r}^2
\label{tempeq:useHoelder1}.
\end{align}
To bound $\abs{v_r}$ we proceed by bounding each of the terms in \eqref{tempeq:useHoelder1}. To this aim, first note that by using the fact that $\kappa=\sqrt{\FK''(0)} = \sqrt{\frac{\pi^2}{3}(\N^2+4\N)}$ (\cite[Eq.~2.23]{candes_towards_2014}) we can conclude that   
\begin{align}
\left(\frac{2\pi \N}{\kappa}\right)^{\onenorm{\vm+\vn}} \leq c',
\label{eq:bound2pinkappa}
\end{align}
holds for $\vm,\vn$ obeying $\onenorm{\vm}, \onenorm{\vn}\le 2$. Using the latter inequality we arrive at
\begin{align}
\label{ftemp}
\infnorm{\tilde{\mathbf{f}}}^2=\left(\frac{\infnorm{\mathbf{f}^\vn(\vr)}}{\kappa^{\onenorm{\vn}}}\right)^2\frac{1}{\L^3}=\left(\frac{\left(2\pi \N\right)^{\onenorm{\vn}}}{\kappa^{\onenorm{\vn}}}\right)^2\frac{1}{\L^3}\le \frac{\tilde{c}}{\L^3}.
\end{align}
Similarly, the definition of $\vg_{\vm}(\vr)$ in \eqref{eq:defvge} together with the coefficients $g_k$ obeying $|g_k|\leq 1$,  implies that $\infnorm{\vg_\vm(\vr_k)}\le \left(2\pi \N\right)^{\onenorm{\vm}}/\M^3$. Thus, using the latter inequality together with inequality~\eqref{eq:bound2pinkappa} and the fact that $\frac{\L}{\M}=\frac{2\N+1}{\N/2+1}\le 4$ yields
\begin{align}
\infnorm{ \tilde \vg}^2 = \L^3\left(\frac{\infnorm{\vg_\vm(\vr_k)}}{\kappa^{\onenorm{\vm}}}\right)^2
\leq
\frac{4^6}{\L^3}\left(\frac{(2\pi \N)^{\onenorm{\vm}}}{\kappa^{\onenorm{\vm}}}\right)^2
\leq 
\frac{\tilde{c}}{\L^3}.
\label{eq:boundtilvg}
\end{align}
Now utilizing \eqref{ftemp} and \eqref{eq:boundtilvg} together with the incoherence condition \eqref{eq:incoherence} in \eqref{tempeq:useHoelder1} yields
\begin{align}
\label{eq:useHoelder1}
\abs{v_r}
\le \frac{c}{\L^3}\onenorm{\va_r}^2\le c\frac{\mu}{\R} =: B.
\end{align}
In order to apply Bernstein's inequality, we also need an upper bound on 
$\sum_{r=\rindstart}^{\rindend{\R}} \EX{ |v_r|^2}$. 
To obtain such a bound, note that 
\begin{align}
\sum_{r=\rindstart}^{\rindend{\R}} \EX{ |v_r|^2}
&\leq 
\sum_{r=\rindstart}^{\rindend{\R}} \EX{ (\herm{\tilde \vf} \va_r \herm{\va}_r \tilde \vg )^2 } \nonumber \\
&\leq 
\sum_{r=\rindstart}^{\rindend{\R}} 
\infnorm{\tilde \vf}
\infnorm{\tilde \vg}  
\EX{\onenorm{\va_r}^2   |\herm{\tilde \vf} \va_r \herm{\va}_r \tilde \vg| }
\label{eq:useHoelder} \\
&\leq
\R \sqrt\frac{\tilde c}{\L^3}
 \sqrt\frac{\bar{c}}{\L^3}
\mu \frac{\L^3}{\R} 
\EX{|\herm{\tilde \vf} \va_r \herm{\va}_r \tilde \vg|}
= 
c' \mu\EX{|\herm{\tilde \vf} \va_r \herm{\va}_r \tilde \vg|}.
\label{eq:defsigmB}
\end{align}
Here, \eqref{eq:useHoelder} follows from H\"older's inequality, and for \eqref{eq:defsigmB} we used the upper bounds \eqref{ftemp} and \eqref{eq:boundtilvg} together with the incoherence assumption  \eqref{eq:incoherence}. Finally note that
\begin{align}
\EX{
\left| \herm{\tilde \vf} \va_r \herm{\va}_r  \tilde \vg \right| 
} 
&\leq 
\EX{ \max \left( \herm{\tilde \vf} \va_r \herm{\va}_r  \tilde \vf , 
\herm{\tilde \vg} \va_r \herm{\va}_r  \tilde \vg  \right)} \nonumber \\
&\leq 
\EX{ \herm{\tilde \vf} \va_r \herm{\va}_r  \tilde \vf + \herm{\tilde \vg} \va_r \herm{\va}_r  \tilde \vg  } \nonumber \\
&\leq 
 \herm{\tilde \vf} \EX{\va_r \herm{\va}_r} \tilde \vf + \herm{\tilde \vg} \EX{ \va_r \herm{\va}_r}  \tilde \vg  \nonumber \\
&=
\frac{1}{\R} \left( \twonorm{\tilde \vf}^2 + \twonorm{\tilde \vg}^2 \right) \label{eq:byisotropy}\\
&\leq 
\frac{c}{\R}
\label{eq:bytftgb}. 
\end{align}
Here, \eqref{eq:byisotropy} follows from the isotropy assumption, 
and for \eqref{eq:bytftgb} we used the upper bounds~\eqref{ftemp} and \eqref{eq:boundtilvg} on the maximal value of $\tilde \vf$ and $\tilde \vg$ together with $\twonorm{\cdot } \leq \sqrt{\L^3} \infnorm{\cdot}$. Plugging \eqref{eq:bytftgb} in \eqref{eq:defsigmB} yields
\begin{align}
\sum_{r=\rindstart}^{\rindend{\R}}
\EX{ |v_r|^2}
\le c\frac{\mu}{\R} =: \sigma^2.
\label{sigma2}
\end{align} 
We now have all the elements to apply Bernstein's inequality. Before applying Bernstein's inequality first note that by assumption~\eqref{eq:coherenceass}, we have
\begin{align}
\frac{c_1}{\sqrt{\S \log ( \L / \delta )   }   }\ge t:=\sqrt{c'\frac{\mu}{\R}\log(\L/\delta)}.
\label{ms}
\end{align}
Thus an application of Bernstein's inequality yields
\begin{align}
\PR{
\frac{1}{\kappa^{\onenorm{\vm + \vn}}} 
\abs{\innerprod{ (\herm{\mA}\mA - \mI) \vg_\vm(\vr_k)}{ \vf^{\vn}(\vr)}}
\geq \frac{c_1}{\sqrt{S \log ( \L / \delta )   }   }}&\le\PR{
\left| \sum_{r=\rindstart}^{\rindend{\R}} v_r  \right|
\geq t}\nonumber\\
&\le 2 \text{exp}\left(-\frac{\frac{1}{2}t^2}{\sigma^2+\frac{1}{3}Bt}\right)\nonumber\\
&\le 2 e^{-15\log(\L/\delta)}\nonumber\\
&\le 2 \frac{\delta}{\L^{15}}.\nonumber
\end{align}
This concludes the proof of condition~\eqref{eq:cond1new}.

We next show that condition~\eqref{eq:GrmGrgA} holds. Note that
\begin{align}
\left | \innerprod{\mA \vg_{\vect{0}}(\vr_k)}{\mA \vf(\vr)} \right|
&\leq \sum_{r=\rindstart}^{\rindend{\R}}
\left| \herm{\vf}(\vr) \va_r \herm{\va}_r \vg_{\vect{0}}(\vr_k) \right| \nonumber \\
&\leq 
\sum_{r=\rindstart}^{\rindend{\R}}
\infnorm{\tilde \vf} \onenorm{\va_r}^2 \infnorm{\tilde \vg}
\leq c \mu
\leq \frac{\R}{\S \log^2(\L/\delta) } 
\leq \frac{\L^3}{\S}. 
\nonumber 
\end{align}
Here, we used incoherence ($\onenorm{\va_r}^2 \leq \frac{\L^3}{\R} \mu$), the upper bounds \eqref{ftemp} and \eqref{eq:boundtilvg}, 
and finally assumption \eqref{eq:coherenceass} and $\R \leq \L^3$. 
This concludes the proof of \eqref{eq:GrmGrgA}. 


\subsection{\label{sec:proofmimo}Proof of Theorem \ref{thm:mainresmimo}: MIMO radar}

\newcommand{\msA}{\bar \mA} 

We first show that the MIMO-radar input-output relation~\eqref{eq:periorelMIMO} obeys $\vy = \mAA \vz$, with $\mAA \in \complexset^{N_R \L \times N_R N_T \L^2}$ defined in~\eqref{eq:defAmimo}.  
Recall that we assume that $\L = N_R N_T$. 
First note that equation~\eqref{eq:periorelMIMO} can be written as 
\begin{align}
\vy
&=
\begin{bmatrix}
\vy_{\rindstart}\\
\vdots \\
\vy_{\rindend{\R}}
\end{bmatrix} 
= 
\sum_{k=\indstart}^{\indend{\S}}
b_k' 
\begin{bmatrix}
e^{i2\pi 0 \beta_k }  \sum_{j=0}^{N_T-1} e^{i2\pi j \beta} \mc F_{\nu_k}
\mc T_{\tau_k} \vx_j  \\
\vdots  \\
e^{i2\pi N_T (N_R-1) \beta_k }  \sum_{j=0}^{N_T-1} e^{i2\pi j \beta_k } \mc F_{\nu_k}
\mc T_{\tau_k} \vx_j 
\end{bmatrix} \nonumber \\
&=
\sum_{k=\indstart}^{\indend{\S}}
\underbrace{b_k' e^{i2\pi \beta_k \N} }_{b_k}
\begin{bmatrix}
[\mc F_{\nu}  \mc T_{\tau_k} \vx_{0} \ldots \mc F_{\nu_k}  \mc T_{\tau_k} \vx_{N_T-1} ] & \vect{0} & \vect{0} \\
 \vect{0}& \ddots & \vect{0} \\
\vect{0} &  \vect{0} & [\mc F_{\nu_k}  \mc T_{\tau_k} \vx_{0} \ldots \mc F_{\nu_k} \mc T_{\tau_k} \vx_{N_T-1} ]
\end{bmatrix}
\vf(\beta),
\label{eq:mimoinmtxform}
\end{align}
where $[\vf(\beta)]_j = e^{i2\pi \beta j }, j =-\N,\ldots,\N$. 
With 
\[
\mc F_{\nu}  \mc T_{\tau} \vx 
=
\mG_{\vx} \herm{\mF} 
\vf([\tau,\nu]),
\]
from~\eqref{eq:tfshiftGF}, it is seen that equality~\eqref{eq:mimoinmtxform} becomes 
\begin{align}
\vy =
\sum_{k=\indstart}^{\indend{\S}} 
b_k \mA \vf(\vr), 
\quad \mA \defeq \begin{bmatrix}
\mB &  &  \\
& \ddots & \\
 & & \mB 
\end{bmatrix}. 
\end{align}
This concludes the proof of the MIMO-radar input-output relation~\eqref{eq:periorelMIMO} obeying relation~$\vy = \mA \vz$ with $\vz$ defended in~\eqref{eq:syseqinwdim}. 

The remainder of the proof of Theorem \ref{thm:mainresmimo} is again  similar to the proof of Theorem \ref{thm:incoherence} in Section \ref{sec:proofthm1}. 
Specifically, we show that the conditions~\eqref{eq:cond1new} and \eqref{eq:GrmGrgA} of 
Lemma~\ref{prop:dualpolynomial} are satisfied,  which in turn implies the existence of an appropriate dual certificate. 
For notational convenience, we prove that those conditions are satisfied for the normalized measurement model
\[
\bar \vy = \msA \vz, 
\text{
where $\bar \vy \defeq \sqrt{\L^2 N_T} \vy$ and $\msA \defeq \sqrt{\L^2 N_T} \mA$. 
}
\]
In a nutshell, we first show that $\innerprod{ \msA \vg_\vm(\vr_k)}{\msA \vf^{\vn}(\vr)}$ in conditions~\eqref{eq:cond1new} and \eqref{eq:GrmGrgA}  can be expressed as the quadratic form (see Section~\ref{sec:proofquadform}) 
\begin{align}
\innerprod{ \msA \vg_\vm(\vr_k)}{\msA \vf^{\vn}(\vr)}
= 
\herm{\vx}  \mV_{\vm}^{(\vn)}(\vr, \vr_k) \, \vx, 
\label{eq:gmnqform}
\end{align}
where $\mV_{\vm}^{(\vn)}(\vr, \vr_k) \in \complexset^{\L N_T \times \L N_T}$ is defined in equation~\eqref{eq:exprVqf} in Section~\ref{sec:proofquadform} below, and $\vx \defeq \transp{[\transp{\vx}_0,\ldots,\transp{\vx}_{N_T-1}]}$ is a Gaussian random vector (recall that the probing signals are Gaussian). 
Conditions~\eqref{eq:cond1new} and \eqref{eq:GrmGrgA} then follow by utilizing the Hanson-Wright inequality stated in the lemma below,
which ensures that this quadratic form does not deviate too much from its expectation
\begin{align}
\EX{\innerprod{ \msA \vg_\vm(\vr_k)}{\msA \vf^{\vn}(\vr)}}
=
\innerprod{\vg_\vm(\vr_k)}{\vf^{\vn}(\vr)}. 
\label{eq:exgmn}
\end{align}
Equality in~\eqref{eq:exgmn} holds by $\EX{\herm{\msA} \msA } = \mI$, which in turn follows from 
$\EX{\herm{\mG}_{\vx_k} \mG_{\vx_k}} = \frac{1}{N_T} \mI$ \cite[Sec.~8.1]{heckel_super-resolution_2015} by noting that 
\[
 \EX{  \herm{(\mG_{\vx_k} \herm{\mF})} \mG_{\vx_k} \herm{\mF} } 
 =
 \mF \, \EX{ \herm{\mG}_{\vx_k} \mG_{\vx_k} } \herm{\mF}
 = 
 \frac{1}{\L^2 N_T} \mI 
\]
and 
\[
\EX{ \herm{(\mG_{\vx_k} \herm{\mF})} \mG_{\vx_{k'} } \herm{\mF}  } = \vect{0},\quad \text{for } k \neq k'. 
\]

\begin{lemma}[Hanson-Wright inequality {\cite[Thm.~1.1]{rudelson_hanson-wright_2013}}]
Let $\mathbf{x} \in \mathbb R^{\L N_T}$ be a random vector with independent zero-mean $K$-sub-Gaussian entries (i.e., the entries obey $\sup_{p\geq 1} p^{-1} (\EX{|x_\ell|^p})^{1/p} \leq K$), and let $\mV$ be an $\L N_T\times \L N_T$ matrix. Then, for all $t\geq 0$, 
\[
\PR{ | \transp{\vx} \mV \mathbf{x}  -  \EX{\transp{\vx} \mV \vx} |  > t } \leq 2 \exp\left( - \tilde c \min\left( \frac{t^2}{K^4 \norm[_F]{\mV}^2  },  \frac{t}{K^2 \norm[]{\mV}  }\right) \right),
\]
where $\tilde c$ is a numerical constant. 
\label{thm:hanswright}
\end{lemma}

We next detail how conditions~\eqref{eq:cond1new} and \eqref{eq:GrmGrgA} follow from the Hanson-Wright inequality. 
In order to evaluate the RHS of the Hanson-Wright inequality, we need the following upper bound on the norm of $\mV_{\vm}^{(\vn)}(\vr, \vr_k)$, proven in Section~\ref{sec:prooflem:boundmV} below. 

\begin{lemma}
\label{lem:boundmV}
For all $\vr, \vr_k \in [0,1]^3$, and for all positive integer vectors $\vn, \vm$ with $\onenorm{\vn + \vm} \leq 4$, 
we have
\begin{align}
\norm[F]{ \mV_{\vm}^{(\vn)}(\vr, \vr_k) }  
\leq 
c_4 
(2\pi \N)^{\onenorm{\vm + \vn}}
N_T  \sqrt{\L}. 
\label{eq:boundonVFnorm}
\end{align}
\end{lemma}

We are now ready to establish \eqref{eq:cond1new} by applying the Hanson-Wright inequality. 
Setting $\mV \defeq \mV_{\vm}^{(\vn)}(\vr, \vr_k)$ for ease of presentation, it follows that 
\begin{align}
&\hspace{-1cm}
\PR{ \frac{1}{\kappa^{\onenorm{\vm + \vn }}}  
\left|\innerprod{ (\herm{\msA}\msA - \mI) \vg_\vm(\vr_k)}{ \vf^{\vn}(\vr)}\right|
> 
\frac{c_1}{
\sqrt{\S \log(\L/\delta)}
} 
}
\nonumber \\
&\leq 
\PR{ \frac{1}{\kappa^{\onenorm{\vm + \vn }}}  
\left|\innerprod{ (\herm{\msA}\msA - \mI) \vg_\vm(\vr_k)}{ \vf^{\vn}(\vr)}\right|
> c_4 c' \frac{\alpha}{\sqrt{\L}}  
} \label{eq:byassumiLS} \\
&\leq
\PR{
\left|\innerprod{ (\herm{\msA}\msA - \mI) \vg_\vm(\vr_k)}{ \vf^{\vn}(\vr)}\right|
> c_4  (2\pi \N)^{\onenorm{\vm + \vn}}   \frac{\alpha}{\sqrt{\L}}  } \label{eq:byconstleks} \\
&\leq \PR{ | \transp{\mathbf{x}} \mV  \mathbf{x}  -  \EX{\transp{\mathbf{x}} \mV  \mathbf{x}} |  >  \norm[F]{ \mV} \frac{\alpha}{\L N_T }  } \label{eq:uselemboundvfnor} \\
&\leq 2 \exp\left( - \tilde c \min\left( \frac{\norm[F]{\mV }^2 \alpha^2}{(\L N_T)^2 K^4 \norm[F]{\mV  }^2  },  \frac{\norm[F]{\mV} \alpha}{ \L N_T K^2 \norm[]{\mV}  }\right) \right) \label{eq:usehansonwr} \\
&\leq 2 \exp\left( - \tilde c \min\left( \frac{ \alpha^2}{c_5^4   },  \frac{ \alpha}{ c_5^2   }\right) \right).  \label{eq:simphanswrer} \\
&\leq
c_3 \frac{\delta}{\L^{15} }  \label{eq:choicealen}
\end{align}
Here, \eqref{eq:byassumiLS} follows by assumption \eqref{eq:sleqb}: 
\begin{align}
\frac{c_1}{
\sqrt{\S \log(\L/\delta)}
} 
\geq  
\frac{c_1}{
\sqrt{\log(\L/\delta)}
} \frac{ \sqrt{c \log^{3}(\L/\delta)} }{\sqrt{\L}}
\geq c_4 c' \frac{\alpha}{\sqrt{\L}},
\end{align}
where we set $\alpha = \frac{c_5^2}{\tilde c} \log\left( \L^{15} /(2c_3\delta) \right)$. 
Equation~\eqref{eq:byconstleks} follows from $(2\pi \N)^{\onenorm{\vm + \vn}} \leq c' \kappa^{\onenorm{\vm + \vn}}$, by \eqref{eq:bound2pinkappa}. 
Moreover, \eqref{eq:uselemboundvfnor} follows from~\eqref{eq:gmnqform}, \eqref{eq:exgmn}, and~\eqref{eq:boundonVFnorm}. 
To obtain \eqref{eq:usehansonwr}, we used Lemma \ref{thm:hanswright} with $t=\norm[F]{\mV} \frac{\alpha}{\L N_T}$. 
Equation~\eqref{eq:simphanswrer} holds because the sub-Gaussian parameter $K$ of the random variable $[\mathbf{x}]_\ell \sim \mathcal N(0,1/(\L N_T))$ is given by $K = c_5/\sqrt{\L N_T}$ (see e.g., \cite[Ex.~5.8]{vershynin_introduction_2012}) and $\norm[F]{\mV }/\norm[]{\mV }  \geq 1$. 
Finally, \eqref{eq:choicealen} follows by using the fact that for $\alpha \geq c_5^2$, we have $\min\left( \frac{ \alpha^2}{c_5^4   },  \frac{ \alpha}{ c_5^2   }\right) = \frac{ \alpha}{ c_5^2}$.

We next verify \eqref{eq:GrmGrgA}. 
We have 
\begin{align}
\PR{
\max_{\vr \in [0,1]^3}  \left | \innerprod{\msA \vg_{\vect{0}}(\vr_k)}{\msA \vf(\vr)} \right| \geq  \frac{\hat c}{S} \L^{3}  
} 
&= 
\PR{
\max_{\vr \in [0,1]^3}  \left| \herm{\vx} \mV_{\vect{0}}^{(\vect{0})}(\vr, \vr_k) \vx \right|
 \geq  \frac{\hat c}{S} \L^{3}   
}\nonumber \\ 
&\leq 
\PR{
\opnorm{ \mV_{\vect{0}}^{(\vect{0})}(\vr, \vr_k) } \twonorm{\vx}^2 
 \geq  \frac{\hat c}{S} \L^{3}  
} \nonumber \\
&\leq 
\PR{
c_4 N_T \sqrt{\L}  \twonorm{\vx}^2 
 \geq  \frac{\hat c}{S} \L^{3}  
} \label{eq:ueq:boundonVFnorm} \\
&\leq 
\PR{
c_4 N_T \sqrt{\L}  \twonorm{\vx}^2 
 \geq  \frac{\hat c \, c \log^3(\L/\delta)  }{\L} \L^{3}  \label{eq:beq:sleqba}
} \\
&\leq \PR{ \twonorm{\vx}^2 
 \geq  c_5 \log^3(\L/\delta)  \sqrt{\L}  
} \nonumber \\
&\leq \frac{\delta}{8}. \label{eq:usechisqconc}
\end{align}
Here \eqref{eq:ueq:boundonVFnorm} follows from \eqref{eq:boundonVFnorm} and $\opnorm{\cdot} \leq \norm[F]{\cdot}$, and \eqref{eq:beq:sleqba} follows by the assumption \eqref{eq:sleqb}. Finally, \eqref{eq:usechisqconc} follows from a standard concentration bound of a $\chi^2$-random variable.

\subsubsection{\label{sec:proofquadform}Proof of Equation~\eqref{eq:gmnqform}}

In this section we prove that 
$\innerprod{ \msA \vg_\vm(\vr_k)}{\msA \vf^{\vn}(\vr)}$ 
can be expressed as the quadratic form in $\vx$. 
To this end, we note that the vector $\vg_\vm(\vr_k)$ defined in \eqref{eq:defvge} can be written as 
\[
\vg_{\vm}(\vr_k)
= 
\vg_{m_1}(\beta_k) \otimes \vg_{m_3}(\nu_k) \otimes \vg_{m_2}(\tau_k)
 \in \complexset^{\L^3}, 
\quad 
\vr_k = \transp{[\tau_k,\nu_k,\beta_k]}, 
\]
where $\otimes$ denotes the Kronecker product, and 
\[
[\vg_{m}(\tau)]_r 
\defeq \frac{1}{\M} g_r  e^{-i2\pi \tau r } (i2\pi r)^m, \quad r = -\N, \ldots, \N. 
\]
As shown the paper~\cite[p.~25]{heckel_super-resolution_2015})
\begin{align}
\hspace{1cm}&\hspace{-1cm}[\L \mG_{\vx_j} \herm{\mF} \vg_{m_3}(\nu) \otimes \vg_{m_2}(\tau) ]_p  \nonumber \\
&=
\frac{1}{\L}
\sum_{k,\ell=-\N}^\N 
\left(
\sum_{r,q=-\N}^\N  
g_r g_q e^{-i2\pi(\tau r + \nu q)}  (i2\pi r)^{m_2} (i2\pi q)^{m_3}
e^{i2\pi \frac{q k + r\ell }{\L}}
\right)
[\vx_j]_{p-\ell} e^{i2\pi \frac{kp}{\L}} 
\nonumber \\
&=
 \sum_{\ell=-\N}^\N  [\vx_j]_{\ell}   
 \underbrace{
 \sum_{r=-\N}^\N e^{i2\pi \frac{r(p-\ell)}{\L}} 
 g_r g_p e^{-i2\pi(\tau r - \nu p)}  (i2\pi r)^{m_2} (-i2\pi p)^{m_3}
}_{ h_{p,\ell} \defeq}
\label{eq:GFHgs},
\end{align}
where we used that $\frac{1}{\L} \sum_{k=-\N}^\N e^{i2\pi \frac{k (p+q)}{\L}}$ is equal to $1$ if $p= -q$ and equal to $0$ otherwise, together with the fact that $[\vx_j]_\ell$ is $\L$-periodic in $\ell$.   
We next write \eqref{eq:GFHgs} in matrix-vector form according to 
\[
\L \mG_{\vx_j} \herm{\mF} (\vg_{m_3}(\nu) \otimes \vg_{m_2}(\tau))  = 
\mH_\vg \vx_j, 
\]
where the $(p,\ell)$-th entry of $\mH_\vg \in \complexset^{\L\times \L}$ is given by $h_{p,\ell}$ (note that $\mH_\vg$ is a function of $\vg_{m_3}(\nu) \otimes \vg_{m_2}(\tau)$). 
Recall that $\vx = \transp{[\transp{\vx}_0,\ldots,\transp{\vx}_{N_T-1}]} \in \complexset^{\L N_T}$. 
With this notation, we have
\begin{align*}
&\L\mA \vg_{\vm}(\vr_j) \\
&= 
\L \begin{bmatrix}
[\mG_{\vx_{0}} \herm{\mF}, \ldots , \mG_{\vx_{N_T - 1 } } \herm{\mF} ] & & \\
& \ddots & \\
& & [\mG_{\vx_{0}} \herm{\mF},\ldots , \mG_{\vx_{N_T - 1 } } \herm{\mF} ] 
\end{bmatrix}
\vg_{m_1}(\beta_j) \otimes \vg_{m_3}(\nu_j) \otimes \vg_{m_2}(\tau_j)   \\
&=
\L
\begin{bmatrix}
\ddots & & \\
& [\ldots  \mG_{\vx_k} \herm{\mF}  (\vg_{m_3}(\nu_j) \otimes \vg_{m_2}(\tau_j)) \ldots ] 
& \\
& & \ddots 
\end{bmatrix}  
\vg_{m_1}(\beta_j) \\
&= 
\begin{bmatrix}
\ddots & & \\
& [\mH_{\vg} \vx_0,\ldots,\mH_{\vg} \vx_{N_T-1}] 
& \\
& & \ddots 
\end{bmatrix}  
\vg_{m_1}(\beta_j)  \\
&=
\mB_{\vg} \vx, 
\end{align*}
where $\mB_{\vg} \in \complexset^{N_R \L \times N_T \L}$ is a block matrix where the $(\ell,k)$-th block, $\ell = -\frac{N_R-1}{2},\ldots,\frac{N_R-1}{2}$, $k = -\frac{N_T-1}{2},\ldots,\frac{N_T-1}{2}$, is given by $e^{i2\pi \beta(k+\ell N_T)}\mH_\vg$. 

Analogously, we define the matrix $\mH_\vf \in \complexset^{\L\times \L}$ such that
\[
\L \mG_{\vx_j} \herm{\mF} \vf^{(m_3)} (\nu) \otimes \vf^{(m_2)}(\tau)
=
\mH_\vf \vx_j,
\]
and a block matrix $\mB_{\vf} \in \complexset^{N_R \L \times N_T \L}$ with $(\ell,k)$-th block given by $e^{i2\pi \beta(k+\ell N_T)}\mH_\vf$, 
such that 
\[
\L \mA \vf^{(\vn)}(\vr)
= \mB_\vf \vx. 
\]
With this notation,
\begin{align*}
&\innerprod{ \msA \vg_\vm(\vr_k)}{\msA \vf^{\vn}(\vr)}
=
N_T \innerprod{ \L \mA \vg_\vm(\vr_k)}{\L\mA \vf^{\vn}(\vr)}
=
\herm{\vx} \mV_{\vm}^{(\vn)}(\vr, \vr_j) \vx, 
\end{align*}
with 
\begin{align}
\mV_{\vm}^{(\vn)}(\vr, \vr_j)
\defeq 
N_T \herm{\mB}_\vf \mB_\vg. 
\label{eq:exprVqf}
\end{align}
This concludes the proof of equation~\eqref{eq:gmnqform}. 

\subsubsection{\label{sec:prooflem:boundmV}Proof of Lemma \ref{lem:boundmV}}

We are now ready to prove Lemma \ref{lem:boundmV}. 
First note that from the previous section, 
\[
\mV_{\vm}^{(\vn)}(\vr, \vr_j) = \frac{N_T}{\M} \mD \otimes \frac{1}{\M^2}\herm{\mH}_\vf \mH_\vg,
\]
with $\mD \in \complexset^{N_T \times N_T}$ defined as
\[
\mD 
\defeq \frac{N_T}{\M} \sum_{\ell= - \sfrac{(N_R - 1)}{2} }^{\sfrac{(N_R - 1)}{2}} 
\begin{bmatrix}
\vdots \\
(i2\pi (k + \ell N_T))^{n_1} e^{i2\pi \beta (k + \ell N_T)} \\
\vdots 
\end{bmatrix}
[\hdots (i2\pi (k' + \ell N_T))^{m_1} g_{k' + \ell N_T} \hdots],
\]
where $k,k' = - (N_T-1)/2,\ldots,(N_T-1)/2$ are the indices of the vectors above. 
With $k+ \ell N_T \leq \N$, and using that $g_k \leq 1$ we have 
\[
\norm[F]{\mD} 
\leq 
\frac{N_T}{ \M } N_R \sqrt{N_T} (2\pi \N)^{n_1}  \sqrt{N_T} (2\pi \N)^{m_1}
\leq 
4
(2\pi \N)^{m_1 + n_1} N_T, 
\]
where we used that $N_R N_T = \L = 2\N + 1 \leq 4 \M$. 
It follows that 
\[
\norm[F]{ \mV_{\vm}^{(\vn)}(\vr, \vr_j) } 
=
\norm[F]{ \frac{N_T}{\M} \mD} \norm[F]{  \frac{1}{\M^2}  \herm{\mH}_\vf \mH_\vg  }  
\leq 
c_1 
(2\pi \N)^{\onenorm{\vm + \vn}}
N_T  \sqrt{\L}, 
\]
which concludes the proof. 
For the inequality above, we used \cite[Lem.~3]{heckel_super-resolution_2015} which states that $\norm[F]{  \frac{1}{\M^2}  \herm{\mH}_\vf \mH_\vg  }\leq c_2 (2\pi \N)^{m_2 + m_3 + n_2 +n_3} \sqrt{\L}$.



\section{\label{GordonProofs}Proof of Theorem~\ref{thm:ISG}} 

\newcommand\func{f}
\newcommand\tangentcone{\mc T}
\newcommand\convset{\mc C}
\newcommand{\lassoest}[3][]{
\ifthenelse{\equal{#1}{}}{ \tilde \vz(#2,#3) }{
\tilde \vz_{#1}(#2,#3)
} }
\newcommand\polar[1]{{#1}^\circ}
\newcommand\f{\nu}
\newcommand\dualnorm[2][\Tnorm]{\ensuremath{{\left\|#2\right\|}_{#1}^{\ast}}}

To prove Theorem \ref{thm:ISG} we need a few definitions. 

Given a convex function $\func$, the \emph{set of descent} of $\func$ at a point $\vx$ is defined as
\begin{align*}
{\cal D}_{\func}(\vx)
\defeq \Big\{\vh:\text{ }\func(\vx+\vh)\le \func(\vx)\Big\}.
\end{align*}
The \emph{tangent cone} $\tangentcone_{\func}(\vx)$ is the conic hull of the descent set. That is, the smallest closed cone $\tangentcone_{\func}(\vx)$ obeying $\mathcal{D}_{\func}(\vx)\subset\tangentcone_{\func}(\vx)$.
The Gaussian width of a set $\convset \subset \reals^n$ is defined as:
\begin{align*}
\omega(\convset)
\defeq
\EX{\underset{\vz\in\convset}{\sup}~\innerprod{\vw}{\vz}},
\end{align*}
where the expectation is taken over $\vw \sim\mathcal{N(}(\vct{0},\mtx{I})$. 
For the remainder of this section, we set $\func(\cdot) = \norm[\setA]{\cdot}$. 

We now state a theorem from \cite{chandrasekaran_convex_2012} characterizing the sample complexity of structured signal recovery problems from linear measurements. 
This theorem is based on Gordon's escape through the mesh lemma \cite{Gor} and a standard Gaussian concentration inequality for Lipschitz functions. 

\begin{theorem}[{\cite[Cor.~3.3]{chandrasekaran_convex_2012}}] Let $\mA\in\reals^{\R \times \L}$ be a random matrix with i.i.d.~$\mathcal{N}(0,1/\R)$ entries. Assume we have $\R$ linear measurements of the form $\vy=\mA\vz$ from an unknown signal $\vz$. Let $\hat \vz$ be a solution to the
atomic norm minimization problem $\AN(\vy)$, i.e., 
\[
\hat{\vz}=\underset{\bar{\vz}}{\arg\min}\quad
\norm[\setA]{\vz}
\quad\emph{subject to}\quad\vy=\mA\bar{\vz}.
\]
Fix $\eta>0$. 
Then, as long as 
\begin{align}
\R \geq \left(\omega(\tangentcone_{\setA}(\vz)\cap\mathbb{S}^{\L-1})+\eta\right)^2 + 1,
\label{msamp}
\end{align}
with probability at least $1-e^{-\frac{\eta^2}{2}}$ the optimization problem $\AN(\vz)$ recovers the unknown signal, i.e.,~$\hat{\vz}=\vz$.
Here, 
$\mathbb{S}^{\L-1}$ is the unit sphere and $\tangentcone_{\setA}(\vz)$ is the tangent cone associated with the atomic norm, i.e., $\tangentcone_{\setA}(\vz) = \tangentcone_{\func}(\vz)$ with $\func(\cdot) = \norm[\setA]{\cdot}$. 
\label{sampthm}
\end{theorem}
As long as 
\begin{align}
\label{mytemp}
\R \geq 4\omega^2(\tangentcone_{\setA}(\vz)\cap\mathbb{S}^{\L-1}),
\end{align} 
setting $\eta=\sqrt{\R}/2$ 
and applying the inequality $(a+b)^2 \leq 2a^2 + 2b^2$,
we have 
\begin{align*}
\R = \frac{\R}{2}+\frac{\R-2}{2}+1 
&\geq 
2\omega^2(\tangentcone_{\setA}(\vz)\cap\mathbb{S}^{\L-1}) + \frac{\R-2}{2}+1 \nonumber \\
&\geq \left(\omega(\tangentcone_{\setA}(\vz)\cap\mathbb{S}^{\L-1})+\frac{\sqrt{\R-2}}{2}\right)^2 + 1
=\left(\omega(\tangentcone_{\setA}(\vz)\cap\mathbb{S}^{\L-1})+\eta\right)^2 + 1.
\end{align*}
Thus, by Theorem~\ref{sampthm} condition~\eqref{mytemp} ensures exact recovery via atomic norm minimization with probability at least $1-e^{-(\R-2)/8}$. 

We next upper bound the Gaussian width using a connection to denoising with $\ell_1$-regularized minimization from \citet{oymak2013sharp}. 
To this end, we define the proximity operator
\begin{align*}
\lassoest[\func]{\vy}{\lambda} = \arg \min_{\bar{\vz}} \frac{1}{2}
\twonorm{\vy-\bar{\vz}}^2 + \lambda \func(\bar{\vz}).
\end{align*}
\citet[Thm.~1.1,Prop.~5.2]{oymak2013sharp} characterize the mean-square distance (the mean square distance of a set $\convset$ is defined as 
$
D(\convset) \defeq \EX[\vw \sim \mc N(0,\mI)]{\min_{\vs\in \convset} \norm[2]{\vw - \vs}^2 }
$)
of the scaled sub-differential $\partial \func(\vz)$ of $\func$ at $\vx$ as
\[
D(\lambda' \partial \func(\vz)) 
= 
\underset{\sigma>0}{\max}
\frac{
\EX[\vw \sim \mc N(0,\mI)]{
\twonorm{\lassoest[\func]{\vz + \sigma \vw}{\sigma \lambda'} - \vz}^2
}}{\sigma^2}.
\]
Due to 
\begin{align}
\label{eq:relDgwidth}
D(\lambda' \partial \func(\vz))
\geq 
D( \text{cone}(\partial \func(\vz)) )
=
D( \polar{\tangentcone}_\func(\vz) )
\geq 
\omega^2( \tangentcone_\func(\vz) \cap \mathbb{S}^{\L-1} ),
\end{align}
this yields an upper bound on the Gaussian width. 
Here, the first inequality and equality, respectively, can be found in \cite[Eq.~(3.5)]{oymak2013sharp} and \cite[p.~18]{oymak2013sharp} (note that $\text{cone}(\partial \func(\vz))$ is a closed, convex cone, and $\polar{\convset}$ denotes the dual cone of $\convset$).

Thus, the condition
\begin{align}
\label{mytemp2}
\R\geq 
4 \,
\min_{\lambda'}\underset{\sigma>0}{\max}
\frac{
\EX[\vx \sim \mc N(0,\mI)]{
\twonorm{\lassoest[\func]{\vz + \sigma \vw}{\sigma \lambda'} - \vz}^2
}}{\sigma^2}
\end{align} 
implies~\eqref{mytemp}, and therefore, 
\eqref{mytemp2} is sufficient to guarantee that the estimate $\hat \vz$ of $\AN(\vy)$ is exact, i.e., $\hat \vz = \vz$ with probability at least $1-e^{-(\R-2)/8}$. 

We now state a lemma that controls the right hand-side of \eqref{mytemp2} when the frequencies of the mixture components are sufficiently separated. The proof of this lemma is essentially a consequence of results established in \cite{tang_near_2015} and is deferred to Section \ref{pfmlem}. 

\begin{lemma}
\label{lem:denoisinglasso}
Let $\vz$ be of the form~\eqref{eq:defz} with the frequencies of the mixtures obeying the minimum separation condition~\eqref{eq:minsepintro1D}. 
Consider the proximity operator with $\func(\cdot) = \norm[\setA]{\cdot}$, defined as
\begin{align}
\label{eq:proxopat}
\lassoest{\vy}{\lambda} = \arg \min_{\bar{\vz}} \frac{1}{2}
\twonorm{\vy-\bar{\vz}}^2 + \lambda \norm[\setA]{\bar{\vz}}.
\end{align}
Let $\lambda = c \sigma \sqrt{\L \log \L}$, where $c$ is a fixed numerical constant. 
Then the proximity operator obeys 
\begin{align}
\frac{\EX[\vw \sim \mc N(0,\mI)]{\norm[2]{ \lassoest{\vz + \sigma \vw}{\lambda}  - \vz}^2 }}{\sigma^2}
\leq
c \S \log \L.
\end{align}
\end{lemma}

Combining Lemma \ref{lem:denoisinglasso} above with \eqref{mytemp2} completes the proof. All that remains is to prove Lemma~\ref{lem:denoisinglasso} which is the subject of Appendix~\ref{pfmlem}.

\section{Acknowledgements}
The work of RH was supported by the Swiss National Science Foundation under grant P2EZP2\_159065. 
Theorems \ref{thm:mainresmimo} and \ref{cor:mimodiscrete} on the MIMO radar problem
and the numerical results in Section~\ref{sec:numres}
were presented at the 2016 IEEE International Symposium on Information Theory (ISIT) \cite{heckel_super-resolution_2016b}.

\printbibliography

\appendices

\section{\label{sec:prooflemdualpoly}Proof of Lemma~\ref{prop:dualpolynomial}}

In this section we prove our main technical result. 
As explained in Section \ref{sec:statementlem1}, we construct $Q$ explicitly according to \eqref{eq:dualpolyorig}. 
We proceed as follows:  
\begin{description}
\item[Step 1:]\label{it:step2} 
We show that, with probability at least $1-\frac{\delta}{2}$ there exists a choice of coefficients $\alpha_k, \alpha_{1k}, and\alpha_{2k}$ and $\alpha_{3k}$ such that 
\begin{align}
Q(\vr_k) = u_k \quad \text{and} \quad
\nabla Q(\vr_k) = \vect{0} \quad  \text{for all}\; \vr_k \in \T.    
\label{eq:interpcondQ}
\end{align}
This is necessary to ensures that $Q$ reaches local maxima at the $\vr_k$, which is required for the interpolation and boundedness condition in \eqref{eq:polyboundedinprop} to hold simultaneously. 

\item[Step 2:]\label{it:step3} 
We conclude the proof by showing that with the coefficients from Step 1, 
with probability at least $1-\frac{\delta}{2}$, 
the polynomial $Q$ with the coefficients from step 1 obeys $\abs{Q(\vr)} < 1$ uniformly for all $\vr \in [0,1]^3 \setminus \T$. 
This is accomplished using an $\epsilon$-net argument:
\begin{description}
\item[Step 2a:]\label{it:step3a} Let $\Omega \subset [0,1]^3$ be a (finite) set of grid points that is sufficiently dense in the $\ell_\infty$-norm. Specifically, we choose the points in $\Omega$ on a rectangular grid such that 
\begin{align}
\max_{\vr \in [0,1]^3} \min_{\vr_g \in \Omega} \infdist{ \vr - \vr_g} \leq \frac{\epsilon}{3 \tilde c \L^{4}},
\label{eq:gridmaxdist}
\end{align}
where $\epsilon \defeq 0.0005$.
The cardinality of the set $\Omega$ is 
\begin{align}
|\Omega| = \left(\frac{3\tilde c \L^{4}}{\epsilon}\right)^3 = c' \L^{12}/\epsilon^3 = \tilde c \L^{12}. 
\label{eq:cardOmega}
\end{align}
For every $\vr\in \Omega$, we show that $Q(\vr)$ is ``close'' to $\bar Q(\vr)$ with high probability. 
\item[Step 2b:]\label{it:step3b} 
We use that $Q(\vr)$ is close to $\bar Q(\vr)$ for all $\vr \in \Omega$ combined with Bernstein's polynomial inequality to conclude that 
$Q(\vr)$ is close to $\bar Q(\vr)$ uniformly for all $\vr \in [0,1]^3$.
\item[Step 2c:]\label{it:step3c} 
Finally, we combine this result with the properties of $Q(\vr)$ established in Section \ref{sec:constdpu} (in particular $\abs{\bar Q(\vr)} < 1$ for all $\vr \in [0,1]^3\setminus \T$) to conclude that $\abs{Q(\vr)} < 1$ holds with high probability uniformly for all  $\vr \in [0,1]^3\setminus \T$. 
\end{description}
\end{description}

For this argument to work, we will use that, by assumption~\eqref{eq:cond1new},  
$G_{\vm}(\vr,\vr_k) = \innerprod{\mA \vg_{\vm}(\vr_k) }{ \mA \vf(\vr) }$ 
concentrates around the deterministic function 
$\bar G^{(\vn)}(\vr-\vr_k) = \innerprod{\vg_{\vm}(\vr_k) }{ \vf(\vr) }$.

\subsection{\label{sec:constdpu}Construction of a certificate with unconstraint, deterministic coefficients}

The first building block of our construction of $Q$ is a \emph{deterministic} 3D trigonometric polynomial 
$
\bar Q(\vr) =  \innerprod{\bar \vq }{\vf(\vr)}
$
that satisfies the interpolation and boundedness conditions~\eqref{eq:polyboundedinprop}, but whose coefficients $\bar \vq$ are \emph{not} constraint to be of the form $\herm{\mAA} \vq$. 
As mentioned before, corresponding constructions for 1D and 2D trigonometric polynomials have been derived by Cand\`es and Fernandez-Granda \cite[Prop.~2.1, Prop.~C.1]{candes_towards_2014}. As remarked in \cite{candes_towards_2014}, those results generalize to higher dimensions, albeit with a change in the numerical constant in the minimum separation condition (i.e., the constant $5$ in \eqref{eq:mindistcond}). 
The corresponding generalization is stated below. 

\begin{proposition}
Let $\T = \{\vr_\indstart, \vr_\indsecond,\ldots,\vr_{\indend{\S}}\} \subset [0,1]^3$ be an arbitrary set of points obeying the minimum separation condition \eqref{eq:mindistcond}. Then, for all $u_k \in \complexset, k=\indstart,\ldots,\indend{\S}$, there exists a trigonometric polynomial $\bar Q(\vr) = \innerprod{\bar \vq }{\vf(\vr)}$ 
obeying the interpolation and boundedness conditions \eqref{eq:polyboundedinprop}. 
\label{prop:3D}
\end{proposition}

We provide a few details of the proof of Proposition~\ref{prop:3D} in order to collect properties of the polynomial $\bar Q$ that are needed later for constructing $Q$; however we omit proofs of the technical details as they are similar to the proof of \cite[Prop.~C.1]{candes_towards_2014}. 

As mentioned in Section~\ref{sec:statementlem1}, the polynomial $\bar Q$ is constructed by interpolating the points $u_k$ as indicated in \eqref{eq:detintpolC}. 
To ensure that $\bar Q(\vr)$ reaches local maxima, which is necessary for the interpolation and boundedness condition in \eqref{eq:polyboundedinprop} to hold simultaneously, we choose the coefficients $\bar \alpha_k, \bar \alpha_{1k}, \bar \alpha_{2k}, \bar \alpha_{3k}$ in \eqref{eq:detintpolC} such that 
\begin{align}
\bar Q(\vr_k) = u_k  \quad \text{and} \quad
\nabla \bar Q(\vr_k) = \vect{0}  \quad  \text{for all}\; \vr_k \in \T.
\label{eq:condQinterpolv}
\end{align}
We first establish that there exist coefficients such that \eqref{eq:condQinterpolv} holds. 

\begin{lemma}
Under the hypothesis of Proposition \ref{prop:3D}, there exist coefficient vectors $\bar \val = \transp{[\bar \alpha_1,\ldots,\bar \alpha_\S]}$ and $\val_\ell = \transp{[\bar \alpha_{\ell1},\ldots,\bar \alpha_{\ell\S}]}$ such that the polynomial $\bar Q(\vr)$ in \eqref{eq:detintpolC} obeys \eqref{eq:condQinterpolv}. Moreover, those coefficients satisfy
\begin{align}
&\infnorm{\bar \val} \leq 1.0021, \quad 
\infnorm{\bar \val_\ell} \leq 0.0132, \quad \ell = 1,2,3, \nonumber \\
&\text{and, if } u_k =1,  \Re{\alpha_k} \geq 0.9868, \; \abs{\Im{\alpha_k}} \leq 0.0132.
\label{eq:infalb}
\end{align}
\label{lem:boundcoeffs}
\end{lemma}
\begin{proof}[Proof outline of Lemma \ref{lem:boundcoeffs}]
 Writing \eqref{eq:condQinterpolv} in matrix form yields 
\begin{align*}
\underbrace{
\begin{bmatrix}
\bar \mD^{(0,0,0)} & \frac{1}{\kappa}  \bar \mD^{(1,0,0)} & \frac{1}{\kappa} \bar \mD^{(0,1,0)} & \frac{1}{\kappa} \bar \mD^{(0,0,1)} \\
-\frac{1}{\kappa} \bar \mD^{(1,0,0)} & -\frac{1}{\kappa^2} \bar \mD^{(2,0,0)} & -\frac{1}{\kappa^2} \bar \mD^{(1,1,0)} & -\frac{1}{\kappa^2} \bar \mD^{(1,0,1)} \\
-\frac{1}{\kappa} \bar \mD^{(0,1,0)} & -\frac{1}{\kappa^2} \bar \mD^{(1,1,0)} & -\frac{1}{\kappa^2}\bar \mD^{(0,2,0)} & -\frac{1}{\kappa^2} \bar \mD^{(0,1,1)} \\
-\frac{1}{\kappa} \bar \mD^{(0,0,1)} & -\frac{1}{\kappa^2} \bar \mD^{(1,0,1)} & -\frac{1}{\kappa^2}\bar \mD^{(0,1,1)} & -\frac{1}{\kappa^2} \bar \mD^{(0,0,2)}
\end{bmatrix}
}_{\bar \mD }
\begin{bmatrix}
\bar \val \\
\kappa \bar \val_1 \\
\kappa \bar \val_2 \\
\kappa \bar \val_3 \\
\end{bmatrix}
&=
\begin{bmatrix}
\vu \\
\vect{0} \\
\vect{0} \\
\vect{0}
\end{bmatrix},
\end{align*}
where $[\bar \mD^{(n_1,n_2,n_3)}]_{j,k} \defeq \bar G^{(n_1,n_2,n_3)}(\vr_j - \vr_k)$ and $\kappa \defeq \sqrt{|\FK^{(2)}(0)|}$. 
Here, we have scaled the entries of $\bar \mD$ such that its diagonal entries are $1$ ($\FK(0)=1$, $\kappa^2 = |\FK^{(2)}(0)|$, and $\FK^{(2)}(0)$ is negative). 
%
The existence of coefficients such that $\bar Q$ obeys \eqref{eq:condQinterpolv} follows from $\bar \mD$ being invertible, which is established next.  
\begin{proposition}
$\bar \mD$ is invertible and obeys
\begin{align}
 \norm[]{\mI  - \bar \mD}  &\leq 0.03254 \label{eq:IDb} \\
 \norm[]{ \inv{\bar \mD}} &\leq 1.03363.  \label{eq:boundinvbard}
\end{align}
\label{prop:invertofbarD}
\end{proposition}
The proof of Proposition \ref{prop:invertofbarD}, not detailed here, shows that $\bar \mD$ is close to the identity matrix $\mI$ by using that $G(\vr - \vr_k)$ and its derivatives decay fast around $\vr_k$, and that the points $\vr_k \in \T$ are sufficiently separated, by the minimum separation condition \eqref{eq:mindistcond}. 

Since $\bar \mD$ is invertible, the coefficients in Lemma \ref{lem:boundcoeffs} are given by:
\begin{align}
\begin{bmatrix}
\bar \val \\
\kappa \bar \val_1 \\
\kappa \bar \val_2 \\
\kappa \bar \val_3 \\
\end{bmatrix}
=
\inv{\bar \mD} 
\begin{bmatrix}
\vu \\
\vect{0} \\
\vect{0} \\
\vect{0}
\end{bmatrix}
=
\bar \mL 
\vu,
\label{eq:alph}
\end{align}
where $\bar \mL$ is the $4\S \times \S$ submatrix of $\inv{\bar \mD}$ corresponding to the first $\S$ columns of $\inv{\bar \mD}$. 
If follows from $\bar \mD$ being close to identity that the $\ell_\infty$-norm of $\bar \mD$ is small, which establishes \eqref{eq:infalb}. 
\end{proof}
 
The proof of Proposition \ref{prop:3D} is concluded by showing that $\bar Q(\vr)$ satisfies the boundedness conditions in \eqref{eq:polyboundedinprop} as well. 
This is formalized by the following lemma. 
\begin{lemma}
Let $Q(\vr)$ be the polynomial in \eqref{eq:detintpolC} with coefficients given by \eqref{eq:alph}. 
Under the hypothesis of Proposition \ref{prop:3D},  
\begin{enumerate}[i.]
\item $|Q(\vr)| < 1$ for all $\vr \in [0,1]^3 \setminus \T$ that satisfy $\min_{\vr_k \in \T} |\vr - \vr_k | \leq \cfar/\N$,
\item for all $\vr$ that satisfy $\min_{\vr_k \in \T} |\vr - \vr_k | \geq \cfar/\N$  
it follows that
$
\abs{Q(\vr)} < 0.99.
$ 
\end{enumerate}
\label{lem:farpoints}
\end{lemma}
%

\subsection{Step 1: Choice of the coefficients of $\bar Q$}

We next show that, with high probability, we can select the coefficients of $Q$ such that it satisfies \eqref{eq:interpcondQ}. 
To this end, we write \eqref{eq:interpcondQ} in matrix form: 
\begin{align}
\underbrace{
\begin{bmatrix}
\mD_{(0,0,0)}^{(0,0,0)} & \frac{1}{\kappa}  \mD_{(1,0,0)}^{(0,0,0)} & \frac{1}{\kappa} \mD_{(0,1,0)}^{(0,0,0)} & \frac{1}{\kappa} \mD_{(0,0,1)}^{(0,0,0)} \\
-\frac{1}{\kappa} \mD^{(1,0,0)}_{(0,0,0)} & -\frac{1}{\kappa^2} \mD^{(1,0,0)}_{(1,0,0)} & -\frac{1}{\kappa^2} \mD^{(1,0,0)}_{(0,1,0)} & -\frac{1}{\kappa^2} \mD^{(1,0,0)}_{(0,0,1)} \\
-\frac{1}{\kappa} \mD^{(0,1,0)}_{(0,0,0)} & -\frac{1}{\kappa^2} \mD^{(0,1,0)}_{(1,0,0)} & -\frac{1}{\kappa^2} \mD^{(0,1,0)}_{(0,1,0)} & -\frac{1}{\kappa^2} \mD^{(0,1,0)}_{(0,0,1)} \\
-\frac{1}{\kappa} \mD^{(0,0,1)}_{(0,0,0)} & -\frac{1}{\kappa^2} \mD^{(0,0,1)}_{(1,0,0)} & -\frac{1}{\kappa^2} \mD^{(0,0,1)}_{(0,1,0)} & -\frac{1}{\kappa^2} \mD^{(0,0,1)}_{(0,0,1)} 
\end{bmatrix}
}_{\mD \defeq}
\begin{bmatrix}
\val \\
\kappa \val_1 \\
\kappa \val_2 \\
\kappa \val_3
\end{bmatrix}
=
\begin{bmatrix}
\vu \\
\vect{0} \\
\vect{0} \\
\vect{0}
\end{bmatrix}, 
\label{eq:syseqorig}
\end{align}
where $[\mD^{(n_1,n_2,n_3)}_{(m_1,m_2,m_3)}]_{j,k} \defeq G^{(n_1,n_2,n_3)}_{(m_1,m_2,m_3)}(\vr_j, \vr_k)$, $[\val]_k \defeq \alpha_k$, $[\val_i]_k \defeq \beta_{ik}, i=1,2,3$. To show that the system of equations \eqref{eq:syseqorig} has a solution, and in turn \eqref{eq:interpcondQ} holds, we will show that, with high probability, $\mD$ is invertible. To this end, we show that 
\[
\mc E_\xi \defeq \{ \norm{ \mD - \bar \mD }  \leq \xi\}
\]
occurs with high probability, and that $\mD$ is invertible on $\mathcal E_\xi$ for all $\xi \in (0,1/4]$. The fact that $\mD$ is invertible on $\mathcal E_\xi$ for all $\xi \in (0,1/4]$ follows with \eqref{eq:IDb} by noting that: 
\[
\norm[]{\mI - \mD} \leq \norm[]{\mD - \bar \mD} + \norm[]{\bar \mD - \mI} \leq \xi + 0.03254 \leq 0.28254. 
\]
Since $\mD$ is invertible, the coefficients of $Q$ can be selected as 
\begin{align}
\begin{bmatrix}
\val \\
\kappa \val_1 \\
\kappa \val_2 \\
\kappa \val_3
\end{bmatrix}
= \inv{\mD} \begin{bmatrix}
\vu \\
\vect{0} \\
\vect{0} \\
\vect{0}
\end{bmatrix}
= \mL \vu ,
\label{eq:alphabeta}
\end{align}
where $\mL$ is the $4\S \times \S$ submatrix of $\inv{\mD}$ consisting of the first $\S$ columns of $\inv{\mD}$. 
In the remainder of this proof, we assume that the coefficients are given by \eqref{eq:alphabeta}. 
We will need the following bounds on the norm of $\mL$ and its deviation from $\bar{\mL}$ later. 
\begin{lemma}
On the event $\mc E_\xi$ with $\xi \in (0,1/4]$ the following identities hold
\begin{align}
\norm[]{\mL} \leq&  2.07 \label{eq:normmLb},\\
\norm[]{\mL-\bar{\mL}}\leq& 2.14 \xi. \label{eq:normLmbL} 
\end{align}
\label{lem:LLdiffb}
\end{lemma}
Lemma \ref{lem:LLdiffb} follows with \eqref{eq:IDb} from matrix inequalities, see \cite[Lem.~4]{heckel_super-resolution_2015} for details. 
In the remainder of this proof, we set
\newcommand\deltad{\delta'}
\begin{align}
\xi = \epsilon c_6  \log^{-1/2}\left( \frac{4|\Omega|}{\delta'} \right), 
\quad \deltad \defeq \delta/176,
\label{eq:choicexi}
\end{align}
with $\Omega$ as defined in Appendix~\ref{sec:prooflemdualpoly}, Step 2. 

It remains to establish that the event $\mc E_\xi$ occurs with high probability.
\begin{lemma}
Under the assumptions of Lemma~\ref{prop:dualpolynomial}, 
\[
\PR{  \mc E_\xi } \geq 1 - \deltad.
\]
\label{lem:preptaubound}
\end{lemma}

\begin{proof}
We will upper-bound $\norm[]{\mD - \bar \mD}$ by upper-bounding the largest entry of $\mD - \bar \mD$. To this end, first note that the entries of $\mD - \bar \mD$ are given by
\[
\frac{1}{\kappa^{\onenorm{\vm+\vn}}} [\mD^{(\vn)}_{\vm} -\bar \mD^{(\vm + \vn)}]_{j,k} 
= 
\frac{1}{\kappa^{\onenorm{\vm+\vn} }} (G^{(\vn)}_{\vm}(\vr_j, \vr_k) - \bar G^{(\vm + \vn)}(\vr_j - \vr_k)),
\] 
for $\infnorm{\vm}, \infnorm{\vn} \leq 1$ and $\onenorm{\vm + \vn} \leq 2$ and for $j,k=\indstart,\ldots,\indend{\S}$. 
We now have 
\begin{align}
\PR{ \norm[]{\mD - \bar \mD} \geq \xi } 
&\leq \PR{ \sqrt{4\S} \max_{j,k,\vm, \vn} \frac{1}{\kappa^{\onenorm{\vm + \vn}}} |[\mD^{(\vn)}_{\vm} -\bar \mD^{(\vm + \vn)}]_{j,k}| \geq \xi } \label{eq:ubmaxopnormin} \\
&\hspace{-0.5cm}\leq \sum_{j,k, \vm, \vn} 
\PR{ \frac{1}{\kappa^{\onenorm{\vm + \vn}}} |[\mD^{(\vn)}_{\vm} -\bar \mD^{(\vm + \vn)}]_{j,k}|  \geq \frac{c_6 \epsilon}{  \sqrt{4 \S  \log(4 |\Omega| / \delta ) }}} \label{eq:ubmaxdmnbdm} \\
&\hspace{-0.5cm}\leq
16 (4S)^2 c_3 \frac{\delta}{\L^{15}} 
\leq \delta'.
\label{eq:evlastexi} 
\end{align}
Here, \eqref{eq:ubmaxopnormin} follows from the fact that $\mD$ and $\bar \mD$ are $4\S\times 4\S$ matrices, \eqref{eq:ubmaxdmnbdm} follows from the union bound, and \eqref{eq:evlastexi} follows from assumption \eqref{eq:cond1new}.
Finally, for the last inequality in \eqref{eq:evlastexi}, we used $\S \leq \L^{3}$ and that $c_3$ is sufficiently small. 
\end{proof}


\subsection{Step 2a: $Q(\vr)$ and $\bar Q(\vr)$ are close on a grid}

We next show that $Q(\vr)$ and $\bar Q(\vr)$ and their partial derivative are close on a set of grid points. This results is used in Step 2b in an $\epsilon$-net argument to prove that $Q(\vr)$ and $\bar Q(\vr)$ are close for all $\vr$. 

\begin{lemma}
\label{lem:diffqmbarqmongrid}
Suppose that $\onenorm{\vn} \leq 2$. 
Under the assumptions of Lemma~\ref{prop:dualpolynomial}, 
\[
\PR{
\max_{\vr \in \Omega} \frac{1}{\kappa^{\onenorm{\vn}}} \left| Q^{(\vn)}(\vr) - \bar Q^{(\vn)}(\vr)  \right| \leq \epsilon 
} \geq 1 - \frac{\delta}{44}.
\]
\end{lemma}

In order to prove Lemma \ref{lem:diffqmbarqmongrid}, 
we first note that the $\vn$-th partial derivative of $Q(\vr)$ (defined by \eqref{eq:dualpolyorig}) after normalization with $1/\kappa^{ \onenorm{\vn} }$ is given by
\begin{align}
\frac{1}{\kappa^{\onenorm{\vn}}} Q^{(\vn)}(\vr) 
&= \herm{(\vv^{(\vn)}(\vr) )} \mL \vu.  \label{eq:Qmninnprodform}
\end{align}
Here, we used that the coefficients are given by \eqref{eq:alphabeta} and we defined
\begin{align*}
\herm{(\vv^{(\vn)} (\vr) )}
\defeq \frac{1}{\kappa^{\onenorm{\vn}}} 
\bigg[
&G^{(\vn)}_{(0,0,0)}(\vr,\vr_\indstart), \ldots, G^{(\vn)}_{(0,0,0)}(\vr,\vr_{\indend{\S}}),
\; \frac{1}{\kappa} G^{(\vn)}_{(1,0,0)}(\vr,\vr_\indstart), \ldots, \frac{1}{\kappa} G^{(\vn)}_{(1,0,0)}(\vr, \vr_\indend{\S}), \\
& \frac{1}{\kappa} G^{(\vn)}_{(0,1,0)}(\vr, \vr_\indstart),\ldots, \frac{1}{\kappa} G^{(\vn)}_{(0,1,0)}(\vr, \vr_\indend{\S}), 
\frac{1}{\kappa} G^{(\vn)}_{(0,0,1)}(\vr, \vr_\indstart),\ldots, \frac{1}{\kappa} G^{(\vn)}_{(0,0,1)}(\vr, \vr_\indend{\S})
\bigg].
\end{align*}
Since 
$
\EX{
G^{(\vn)}_{\vm}(\vr, \vr_j)} 
=   \bar G^{(\vn + \vm)}(\vr - \vr_j)
$ (by~\eqref{eq:exgmn}), we have 
\[
\EX{\vv^{(\vn)} (\vr) } = \bar \vv^{(\vn)}(\vr),
\]
where 
\begin{align*}
&\herm{(\bar\vv^{(\vn)}(\vr) )} \!\!\defeq \!\! \frac{1}{\kappa^{\onenorm{\vn}}} 
\bigg[
\bar G^{(\vn)}(\vr-\vr_\indstart), \ldots, \bar G^{(\vn)}(\vr-\vr_\indend{\S}),
\; \frac{1}{\kappa} \bar G^{(n_1+1,n_2,n_3)}(\vr - \vr_\indstart) ,\ldots, \frac{1}{\kappa} \bar G^{(n_1+1,n_2,n_3)}(\vr - \vr_\indend{\S}), \\
&\frac{1}{\kappa} \bar G^{(n_1,n_2+1,n_3)}(\vr - \vr_\indstart) ,\ldots, \frac{1}{\kappa} \bar G^{(n_1,n_2+1,n_3)}(\vr - \vr_\indend{\S}),
\frac{1}{\kappa} \bar G^{(n_1,n_2,n_3+1)}(\vr - \vr_\indstart) ,\ldots, \frac{1}{\kappa} \bar G^{(n_1,n_2,n_3+1)}(\vr - \vr_\indend{\S})
\bigg].
\end{align*}
Next, we decompose the derivatives of $Q(\vr)$ according to 
\begin{align}
\frac{1}{\kappa^{\onenorm{\vn}}} Q^{(\vn)}(\vr)
&=
\innerprod{\vu}{\herm{\mL}  \vv^{(\vn)}(\vr) } \nonumber \\
&= \innerprod{\vu}{  \herm{\bar \mL}  \bar \vv^{(\vn)}(\vr) }
+ \underbrace{\innerprod{\vu}{\herm{\mL} (\vv^{(\vn)}(\vr)  - \bar \vv^{(\vn)}(\vr)) } }_{I^{(\vn)}_1(\vr)}
+ \underbrace{\innerprod{\vu}{\herm{(\mL - \bar \mL)} \bar \vv^{(\vn)}(\vr) } }_{I^{(\vn)}_2(\vr)} \nonumber \\
&= \frac{1}{\kappa^{ \onenorm{\vn} }} \bar Q^{(\vn)}(\vr) + I_1^{(\vn)}(\vr) + I_2^{(\vn)}(\vr), \nonumber 
\end{align}
where $\bar \mL$ was defined in \eqref{eq:alph}. 
It follows that 
\begin{align}
\PR{\max_{\vr \in \Omega} \frac{1}{\kappa^{ \onenorm{\vn} }} \left| Q^{(\vn)}(\vr) - \bar Q^{(\vn)}(\vr)  \right| \geq \epsilon } 
&=
\PR{ \max_{\vr \in \Omega} \left|  I_1^{(\vn)}(\vr) + I_2^{(\vn)}(\vr) \right| \geq \epsilon } \nonumber \\
&\hspace{-4cm}\leq
\PR{ \max_{\vr \in \Omega} \left|  I_1^{(\vn)}(\vr) \right| \geq \epsilon/2 }
+
\PR{  \comp{\mc E}_\xi }
+
\PR{ \max_{\vr \in \Omega} \left| I_2^{(\vn)}(\vr) \right| \geq \epsilon/2 | \mc E_\xi } \label{eq:uppbpri1I2}\\
&\hspace{-4cm}\leq 4 \delta' = \frac{\delta}{44},  \label{eq:maxromgeep}
\end{align}
which concludes the proof of Lemma \ref{lem:diffqmbarqmongrid}. 
Here, \eqref{eq:uppbpri1I2} follows from the union bound and the fact that $\PR{A} =  \PR{A \cap \comp{B}} + \PR{A \cap B} \leq \PR{\comp{B}} + \PR{A | B}$ with 
$B = \mc E_\xi$ and $A = \left\{\max_{\vr \in \Omega} \left| I_2^{(\vn)}(\vr) \right| \geq \epsilon/2 \right\}$. Equation~\eqref{eq:maxromgeep} follows from 
\begin{align}
\PR{\max_{\vr \in \Omega}  |I^{(\vn)}_1(\vr)| \geq \epsilon/2} 
&\leq \deltad +\PR{\comp{ \mc E}_{1/4}},
\label{lem:uboundI1} \\
\PR{\max_{\vr \in \Omega}  |I^{(\vn)}_2(\vr)| \geq \epsilon/2 \Big| \mc E_\xi } 
&\leq \deltad,
\label{lem:uboundI2}
\end{align}
and Lemma \ref{lem:preptaubound} (note that $\xi \leq 1/4$ since $\epsilon/2 \leq 1$ and $c_6\leq 1/4$, thus, Lemma \ref{lem:preptaubound} yields $\PR{ \comp{\mc E}_\xi } \leq \deltad$ and $\PR{  \comp{\mc E}_{1/4} } \leq \deltad$). 
Inequalities~\eqref{lem:uboundI1} and \eqref{lem:uboundI2}, proven
below 
 ensure that the perturbations $I_1^{(\vn)}(\vr)$ and $I_2^{(\vn)}(\vr)$ are small on a set of (grid) points $\Omega$, with high probability.



\subsubsection*{Proof of \eqref{lem:uboundI1} 
}

Set $\Delta \vv^{(\vn)} \defeq \vv^{(\vn)} (\vr) - \bar \vv^{(\vn)} (\vr)$ and $\omega \defeq \frac{\epsilon}{20 \sqrt{\log(8|\Omega|/\deltad)}}$  for notational convenience. We have
\begin{align}
&\PR{\max_{\vr \in \Omega}  |I^{(\vn)}_1(\vr)| \geq \epsilon/2} 
%
=\PR{\max_{\vr \in \Omega} \left|\innerprod{\vu}{\herm{\mL} \Delta \vv^{(\vn)} }\right| \geq 5 \omega \cdot 2 \sqrt{\log(8|\Omega|/\deltad)} } \nonumber \\
&\leq \PR{\bigcup_{\vr \in \Omega } 
\left\{ \left|\innerprod{\vu}{\herm{\mL} \Delta \vv^{(\vn)} }\right| \geq \twonorm{\herm{\mL} \Delta \vv^{(\vn)}} 2 \sqrt{\log(8|\Omega|/\deltad)} \right\} 
\cup \left\{  \twonorm{\herm{\mL} \Delta \vv^{(\vn)}}\geq  5 \omega \right\}  } \nonumber \\
&\leq \PR{\bigcup_{\vr \in \Omega } 
\left\{ \left|\innerprod{\vu}{\herm{\mL} \Delta \vv^{(\vn)} }\right| \geq \twonorm{\herm{\mL} \Delta \vv^{(\vn)}} 2 \sqrt{\log(8|\Omega|/\deltad)} \right\} 
\cup \left\{  \twonorm{\Delta \vv^{(\vn)}}\geq 2\omega \right\} \cup \left\{ \norm[]{\mL} \geq 2.5 \right\}  }  \nonumber  \\
&\leq \PR{\norm[]{\mL} \geq 2.5 } 
+\sum_{\vr \in \Omega} \left( \PR{\left|\innerprod{\vu}{\herm{\mL} \Delta \vv^{(\vn)} }\right| \geq \twonorm{\herm{\mL} \Delta \vv^{(\vn)}} 2 \sqrt{\log(8|\Omega|/\deltad)} }
 +\PR{\twonorm{\Delta \vv^{(\vn)}}\geq  2\omega} \right)  \nonumber \\
&\leq \PR{\comp{ \mc E}_{1/4}} +  |\Omega|  4 e^{-  \log(8|\Omega|/\delta') }
 + \sum_{\vr \in \Omega}  \PR{\twonorm{\Delta \vv^{(\vn)}}\geq 2\omega} \label{eq:useHoeffanddf} \\
&\leq \PR{\comp{ \mc E}_{1/4}} + \frac{\deltad}{2} + \frac{\deltad}{2},  \label{eq:inposltdel}
\end{align}
where \eqref{eq:useHoeffanddf} follows from application of Hoeffding's inequality (stated below, cf.~Lemma \ref{thm:hoeff}) and from $\{\norm[]{\mL} \geq 2.5\} \subseteq \comp{ \mc E}_{1/4}$ according to \eqref{eq:normmLb}.  For \eqref{eq:inposltdel}, we used
\begin{align}
\PR{\twonorm{\Delta \vv^{(\vn)}}^2\geq  4\omega^2} 
\leq
\sum_k \PR{ [\Delta \vv^{(\vn)}]_k^2\geq  \omega^2/S} 
&=  
\sum_k \PR{ [\Delta \vv^{(\vn)}]_k\geq  
\frac{\epsilon}{20 \sqrt{\S \log(8|\Omega|/\deltad)}}
} \nonumber \\
&\leq   
\sum_k \PR{ [\Delta \vv^{(\vn)}]_k\geq  \frac{c_1}{\sqrt{\S \log(\L/\deltad)}} } \label{eq:choapino} \\
&\leq 
4 S \frac{c_3\deltad}{\L^{15} }
\leq  \frac{\deltad}{2 \tilde c \L^{12} } = \frac{\deltad}{2 |\Omega|} ,
\label{eq:byassagil}
\end{align}
where \eqref{eq:choapino} follows from 
$\frac{\epsilon}{20\sqrt{\log(8 |\Omega|/\deltad)}} 
=
\frac{\epsilon}{20\sqrt{\log(8 \tilde c \L^{12}/\deltad)}} 
\geq 
\frac{\epsilon}{20\sqrt{12\log( (8 \tilde c)^{1/12} \L/\deltad)}} 
\geq 
\frac{c_1}{\sqrt{\log(\L/\deltad)}} 
$, provided that $1 \geq (8 \tilde c)^{1/12}$, and 
$\frac{\epsilon}{20\sqrt{12}} \geq c_1$, where $c_1$ is the constant in \eqref{eq:cond1new}. 
Moreover, \eqref{eq:byassagil} follows from the assumption \eqref{eq:cond1new}, 
with $c_3\leq \frac{1}{8\tilde c}$. 

\begin{lemma}[Hoeffding's inequality]
Suppose the entries of $\vu \in \reals^{\S}$ are sampled independently from symmetric distributions on the complex unit disc. Then, for all $t\geq 0$, and for all $\vv \in \complexset^\S$
\[
\PR{ \left|\innerprod{\vu}{\vv}\right|  \geq \twonorm{\vv}  t } \leq 4 e^{- \frac{t^2}{4}}. 
\]
\label{thm:hoeff}
\end{lemma}
\subsubsection*{Proof of \eqref{lem:uboundI2} 
}
By the union bound 
\begin{align}
\PR{\max_{\vr \in \Omega}  |I^{(\vn)}_2(\vr)| \geq \epsilon/2 \Big| \mc E_\xi } 
&\leq \sum_{\vr \in \Omega}
\PR{ \left|\innerprod{ \vu }{\herm{(\mL - \bar \mL)} \bar \vv^{(\vn)}(\vr) }\right| \geq \epsilon/2  \Big| \mc E_\xi } \nonumber \\
&\leq \sum_{\vr \in \Omega}  \PR{\left|\innerprod{\vu}{\herm{(\mL - \bar \mL)} \bar \vv^{(\vn)}(\vr) }\right| \geq  \twonorm{\herm{(\mL - \bar \mL)} \bar \vv^{(\vn)}(\vr) }  \frac{\epsilon c_6}{\xi}  }
\label{eq:useeq:ubltlvr}\\
&\leq |\Omega|  4 e^{- \frac{(\epsilon c_6/\xi)^2}{4}}
\label{eq:useHoeffanddf2} \\
&\leq \delta',  \label{eq:iqledeluc}
\end{align}
where \eqref{eq:useeq:ubltlvr} follows from \eqref{eq:ubltlvr} below,  \eqref{eq:useHoeffanddf2} follows by Hoeffding's inequality (Lemma \ref{thm:hoeff}), and \eqref{eq:iqledeluc} holds by definition of $\xi$ in \eqref{eq:choicexi}. 

To complete the proof, note that by \eqref{eq:normLmbL} we have $\norm[]{\mL - \bar \mL} \leq 2.14 \xi$ on $\mc E_\xi$. Thus, conditioned on $\mc E_\xi$,
\begin{align}
\twonorm{\herm{(\mL - \bar \mL)} \bar \vv^{(\vn)}(\vr) } 
\leq \norm[]{\mL - \bar \mL} \twonorm{ \bar \vv^{(\vn)}(\vr) } \leq 2.14 \xi \onenorm{ \bar \vv^{(\vn)}(\vr) }
\leq c_5 \xi
\leq \frac{\xi}{2c_6} ,
\label{eq:ubltlvr}
\end{align}
where we used $\twonorm{\cdot} \leq \onenorm{\cdot}$, and the third inequality follows from the fact that, for all $\vr$, 
\begin{align*}
\onenorm{ \bar \vv^{(\vn)}(\vr) }
&= \frac{1}{\kappa^{\onenorm{\vn}}} \sum_{k=1}^\S \left( 
\left|\bar G^{(\vn)}(\vr-\vr_k)\right|
+ \left| \frac{1}{\kappa} \bar G^{(n_1+1,n_2,n_3)}(\vr - \vr_k) \right|
+ \left|\frac{1}{\kappa} \bar G^{(n_1,n_2+1,n_3)}(\vr - \vr_k) \right| \right. \\
&\hspace{2cm}
\left. 
+ \left|\frac{1}{\kappa} \bar G^{(n_1,n_2 ,n_3+1)}(\vr - \vr_k) \right|
\right)
\leq \frac{c_5}{2.14}.
\end{align*}
Here, $c_5$ is a numerical constant;
the inequality can be shown analogously to corresponding bounds in \cite[Appendix C.2]{candes_towards_2014} pertaining to the 2D case. 
To prove those bounds we use the minimum separation condition and that the functions $\bar G^{(\vn)}(\vr - \vr_k)$ decay quickly around $\vr_k$. 
Finally, the last inequality in \eqref{eq:ubltlvr} follows from choosing $c_6$ sufficiently small.

\newcommand\x{\tau}
\newcommand\y{\nu}

\subsection{Step 2b: $Q(\vr)$ and $\bar Q(\vr)$ are close for all $\vr$}

We next use an $\epsilon$-net argument together with Lemma \ref{lem:diffqmbarqmongrid} to establish that $Q^{(\vn)}(\vr)$ is close to $\bar Q^{(\vn)}(\vr)$ with high probability uniformly for all $\vr \in [0,1]^3$. 
\begin{lemma}
With probability at least $1 - \frac{\delta}{2}$, 
\begin{align}
\max_{\vr \in [0,1]^3, \;\vn \colon \onenorm{\vn} \leq 2}
\frac{1}{\kappa^{ \onenorm{\vn} }} \left| Q^{(\vn)}(\vr) - \bar Q^{(\vn)}(\vr) \right| 
\leq \epsilon.
\label{eq:ledffbaere}
\end{align}
\label{lem:lemdffqbarqevry}
\end{lemma}
\begin{proof}
First, we use Lemma \ref{lem:diffqmbarqmongrid} to show that 
$\left| Q^{(\vn)}(\vr_g) - \bar Q^{(\vn)}(\vr_g) \right|$ is small for all $\vr_g \in \Omega$. 
Using the union bound over all $11$ vectors $\vn$ with non-negative integer entries obeying $\onenorm{\vn} \leq 2$, it  follows from Lemma \ref{lem:diffqmbarqmongrid}, that
\begin{align}
\left\{
\max_{\vr_g \in \Omega, \vn\colon  \onenorm{\vn} \leq 2}  \frac{1}{\kappa^{\onenorm{\vn} }} \left| Q^{(\vn)}(\vr_g) - \bar Q^{(\vn)}(\vr_g) \right| 
\leq \frac{\epsilon}{3}
\right\} 
\label{eq:QrmgbarQrg}
\end{align}
holds with probability at least $1- \frac{\delta}{4}$. 

We will show in Section \ref{sec:techres2} below that \eqref{eq:QrmgbarQrg} together with 
\begin{align}
\left\{
\max_{\vr \in [0,1]^3,\; \vn\colon \onenorm{\vn} \leq 2}\frac{1}{\kappa^{\onenorm{\vn}} } \left | Q^{(\vn)}(\vr) \right| \leq  \frac{\tilde c}{3} \L^{3}  
\right\}
\label{eq:QrmQrg}
\end{align}
implies \eqref{eq:ledffbaere}. 
In Section \ref{sec:techres1} below we show that \eqref{eq:QrmQrg} holds with probability at least $1-\frac{\delta}{4}$. 
By the union bound, the events in \eqref{eq:QrmgbarQrg} and \eqref{eq:QrmQrg} hold simultaneously with probability at least $1-\frac{\delta}{2}$, which concludes the proof.

\subsubsection{\label{sec:techres1} Proof of the fact that \eqref{eq:QrmQrg} holds with probability at least $1-\frac{\delta}{4}$:} 

We start by upper-bounding $|Q^{(\vn)}(\vr)|$. By \eqref{eq:Qmninnprodform},
\begin{align}
\frac{1}{\kappa^{ \onenorm{\vn} }} \left | Q^{(\vn)}(\vr) \right|   
&= \left| \innerprod{ \mL \vu }{ \vv^{(\vn)}(\vr)}  \right| \nonumber \\
&\leq  \norm[]{{\mL}} \twonorm{ \vu }  \twonorm{ \vv^{(\vn)} (\vr)} \nonumber \\
%
%
&\leq \norm[]{{\mL}} \sqrt{\S}   \sqrt{4\S}\infnorm{ \vv^{(\vn)} (\vr)} \nonumber \\
&= \norm[]{{\mL}} \sqrt{4} \, \S \max_{j,\, \vm \in \{ (0,0,0), (1,0,0), (0,1,0), (0,0,1) \}} 
\frac{1}{\kappa^{ \onenorm{\vn + \vm}}} \left | G^{(\vn)}_{\vm}(\vr,\vr_j) \right| \nonumber \\
&\leq \norm[]{{\mL}} \sqrt{4} \, \S 
\frac{\tilde c}{S} \L^{3}
\label{eq:ubqmnarradf} \\
&\leq 2.5 \sqrt{4} \, \tilde c \L^{3}
\label{eq:ubqmnarr}
\end{align}
where we used $\twonorm{\vu} = \sqrt{\S}$, since the entries of $\vu$ are on the unit disc; 
\eqref{eq:ubqmnarradf} holds with probability at least $1-\delta/8$ according to \eqref{eq:GrmGrg} below, and \eqref{eq:ubqmnarr} holds with probability at least $1-\delta/8$ 
due to  
$
\PR{\norm[]{{\mL}} \geq 2.5} \leq \PR{\comp{\mathcal E}_{1/4} } \leq \frac{\delta}{8}
$ (by \eqref{eq:normmLb} and application of Lemma \ref{lem:preptaubound}). This concludes the proof of \eqref{eq:QrmQrg} holding with probability at least $1-\frac{\delta}{4}$.

It is left to prove that
\begin{align}
\left\{
\max_{\vr,\vr_j \in [0,1]^3, \vn,\vm\colon \onenorm{\vn+\vm} \leq 3}
\frac{1}{\kappa^{\onenorm{\vn+\vm}} } \left | G^{(\vn)}_{\vm}(\vr,\vr_j) \right| 
\leq  \frac{\tilde c}{S} \L^{3}  
\right\}
\label{eq:GrmGrg}
\end{align}
holds with probability at least $1-\delta/8$. 
To this end, note that $G^{(\vn)}_{\vm}(\vr,\vr_j)$ is a trigonometric polynomial in $\vr,\vr_j$. Thus, by Bernstein's inequality, we have
\[
\max_{\vr, \vr_j}
\frac{1}{\kappa^{\onenorm{\vn+\vm}} } \left | G^{(\vn)}_{\vm}(\vr,\vr_j) \right| 
\leq 
\frac{(2\pi \N)^{\onenorm{\vn+\vm}}}{\kappa^{\onenorm{\vn+\vm}} } 
\max_{\vr, \vr_j}
\left | G^{(\vect{0})}_{\vect{0}}(\vr,\vr_j) \right| 
\leq c' \max_{\vr, \vr_j}
\left | G^{(\vect{0})}_{\vect{0}}(\vr,\vr_j) \right| 
\leq 
\frac{\tilde c}{S} \L^{3}
\]
where the last inequality holds with probability at least $1-\delta/8$, by assumption \eqref{eq:GrmGrgA}.

\subsubsection{\label{sec:techres2} Proof of the fact that \eqref{eq:QrmgbarQrg} and \eqref{eq:QrmQrg} imply \eqref{eq:ledffbaere}:}
Consider a point $\vr \in [0,1]^3$ and let $\vr_g$ be the point in $\Omega$ closest to $\vr$ in $\ell_\infty$-distance. By the triangle inequality,
\begin{align}
&\frac{1}{\kappa^{ \onenorm{\vn} }} \left| Q^{(\vn)}(\vr) - \bar Q^{(\vn)}(\vr) \right| 
\leq \nonumber \\
&\hspace{0.5cm}\frac{1}{\kappa^{ \onenorm{\vn} }} \left[ 
\left| Q^{(\vn)}(\vr) - Q^{(\vn)}(\vr_g) \right| 
+ \left| Q^{(\vn)}(\vr_g) - \bar Q^{(\vn)}(\vr_g) \right| 
+ \left| \bar Q^{(\vn)}(\vr_g) - \bar Q^{(\vn)}(\vr) \right|
\right].
\label{eq:Qmqgrid}
\end{align}
We next upper-bound the terms in \eqref{eq:Qmqgrid} separately. With a slight abuse of notation, we write $Q^{(\vn)}(\beta, \tau,\nu) = Q^{(\vn)}(\transp{[\beta,\tau,\nu]}) \allowbreak = Q^{(\vn)}(\vr)$. The first absolute value in \eqref{eq:Qmqgrid} can be upper-bounded according to
\begin{align}
\left| Q^{(\vn)}(\vr) - Q^{(\vn)}(\vr_g) \right| 
&\leq \left| Q^{(\vn)}(\beta,\tau,\nu) - Q^{(\vn)}(\beta,\tau,\nu_g) \right| 
+ \left| Q^{(\vn)}(\beta,\tau,\nu_g)    - Q^{(\vn)}(\beta,\tau_g,\nu_g) \right| 
\nonumber \\
&\quad+ \left| Q^{(\vn)}(\beta,\tau_g,\nu_g) -Q^{(\vn)}(\beta_g,\tau_g,\nu_g) 
\right| \nonumber \\
&\leq |\beta - \beta_g| \sup_{z} \left|Q^{(n_1+1,n_2,n_3)}(z,\tau,\nu)\right|
+ |\tau - \tau_g| \sup_{z} \left|Q^{(n_1,n_2+1,n_3)}(\beta,z,\nu)\right|
\nonumber \\
&\quad+ |\nu - \nu_g| \sup_{z} \left|Q^{(n_1,n_2,n_3+1)}(\beta,z,\nu)\right|
\nonumber \\
&\leq 
|\beta- \beta_g| 2 \pi \N \sup_{z}  \left|Q^{(\vn)}(z,\tau,\nu)\right|
+|\tau- \tau_g| 2 \pi \N \sup_{z}  \left|Q^{(\vn)}(\beta,z,\nu)\right| 
\nonumber \\
&\quad+|\nu- \nu_g| 2 \pi \N \sup_{z}  \left|Q^{(\vn)}(\beta,\tau,z)\right|, 
\label{eq:ubqrmqrg}
\end{align}
where \eqref{eq:ubqrmqrg} follows from Bernstein's polynomial inequality, stated below (recall that $Q^{(\vn)}(\beta,\tau,\nu)$ is a trigonometric polynomial of degree $\N$ in $\beta,\tau$ and $\nu$). 
\begin{proposition}[Bernstein's polynomial inequality {\cite[Cor.~8]{harris_bernstein_1996}
}]
Let $p(\theta)$ be a trigonometric polynomial of degree $\N$ with complex coefficients $p_k$, i.e., $p(\theta) = \sum_{k=-\N}^{\N}  p_k e^{i2\pi \theta k}$ where $p_k$ are complex coefficients. 
Then 
\[
\sup_{\theta} \left| \frac{d}{d\theta} p(\theta)  \right|  \leq 2 \pi \N \sup_{\theta} |p(\theta)|.
\]
\label{prop:bernstein}
\end{proposition}
%
%
On the event \eqref{eq:QrmQrg} we have, by \eqref{eq:ubqrmqrg}, that
\begin{align}
\frac{1}{\kappa^{ \onenorm{\vn} }} \left| Q^{(\vn)}(\vr) - Q^{(\vn)}(\vr_g) \right| 
\leq  \frac{\tilde c}{3} \L^{4} 
( |\beta + \beta_g| + |\tau- \tau_g| + |\nu - \nu_g|)
\leq
\tilde c  \L^{4}
\infdist{\vr - \vr_g}
\leq \frac{\epsilon}{3}, 
\label{eq:diffQrQrg}
\end{align}
where the last inequality follows from \eqref{eq:gridmaxdist}. 

We next upper-bound the third absolute value in \eqref{eq:Qmqgrid}. 
Using steps analogous to those leading to \eqref{eq:diffQrQrg}, 
we obtain  
\begin{align}
\frac{1}{\kappa^{m+n}} \left| \bar Q^{(m,n)}(\vr_g) - \bar Q^{(m,n)}(\vr) \right| 
\leq \frac{\epsilon}{3}. 
\label{eq:barQrmbarQrg}
\end{align}
%
Substituting \eqref{eq:QrmgbarQrg}, \eqref{eq:diffQrQrg}, and \eqref{eq:barQrmbarQrg} into \eqref{eq:Qmqgrid} yields that
\[
\frac{1}{\kappa^{ \onenorm{\vn} }} \left| Q^{(\vn)}(\vr) - \bar Q^{(\vn)}(\vr) \right|  \leq \epsilon,  \text{ for all } \vn\colon \onenorm{\vn} \leq 2 \text{ and for all } \vr \in [0,1]^2.
\]
This concludes the proof. 
\end{proof}


\subsection{Step 2c: $\abs{Q(\vr)} < 1$ for all $\vr \in [0,1]^3 \setminus \T$}

\begin{lemma}
With probability at least $1 - \frac{\delta}{2}$ the following statements hold simultaneously:
\begin{enumerate}[i.]
\item \label{it:stat1} For all $\vr$, that satisfy $\min_{\vr_j \in \T} \infdist{\vr - \vr_j } \geq \cfar/\N$ 
we have that 
$
\abs{Q(\vr)} < 0.9995.
$
\item \label{it:stat2} For all $\vr \notin \T$ that satisfy $0 < \infdist{\vr - \vr_j} \leq \cfar/\N$ for some $\vr_j \in \T$, we have that $\abs{Q(\vr)} < 1$. 
\end{enumerate}
\end{lemma}
\begin{proof}
By Lemma \ref{lem:lemdffqbarqevry}, (recall that $\epsilon =0.0005$), with probability $\geq 1- \frac{\delta}{2}$, 
\begin{align}
\frac{1}{\kappa^{ \onenorm{ \vn } }} \left| Q^{(\vn)}(\vr) - \bar Q^{(\vn)}(\vr) \right| 
\leq 0.0005,
\label{eq:conddiffQbarQinfipr}
\end{align}
 for all nonnegative integer vectors $\vn\colon \onenorm{\vn} \leq 2$, and for all $\vr \in [0,1]^3$. 
Statement \ref{it:stat1} follows directly by combining \eqref{eq:conddiffQbarQinfipr} with Lemma \ref{lem:farpoints} via the triangle inequality. 
 

We next show that Statement \ref{it:stat2} follows from \eqref{eq:conddiffQbarQinfipr} and certain properties of $\bar Q$. 
We assume without loss of generality that $\vect{0} \in \T$, and that $Q(\vect{0}) = 1$, and consider a vector $\vr$ with $|\vr| \leq \cfar/\N$. 
Statement \ref{it:stat2} now follows by showing that the Hessian matrix $\tilde \mH$ of $\tilde Q(\vr) \defeq |Q(\vr)|$ is negative definite. 
By Sylvester's criterion \cite[Thm.~7.2.5]{horn_matrix_2012}, $\tilde \mH$ is negative definite if the leading principal minors denoted by $\tilde \mH_1,\tilde \mH_2$, and $\tilde \mH_3 = \tilde \mH$, have alternating sign, which is implied by
\begin{align}
\det(\tilde \mH_1)&=\tilde Q^{(2,0,0)}(\vr) < 0  \label{eq:minor1} \\
\det(\tilde \mH_2)&=\tilde Q^{(2,0,0)}(\vr) \tilde Q^{(0,2,0)}(\vr)  - \tilde Q^{(1,1,0)}(\vr) \tilde Q^{(1,1,0)}(\vr)
  > 0   \label{eq:minor2}\\
\det(\tilde \mH) &=\tilde Q^{(2,0,0)}(\vr) \tilde Q^{(0,2,0)}(\vr) \tilde Q^{(0,0,2)}(\vr)
+  2 \tilde Q^{(1,1,0)}(\vr) \tilde Q^{(0,1,1)}(\vr) \tilde Q^{(1,0,1)}(\vr) \nonumber \\
&-  \tilde Q^{(1,0,1)}(\vr) \tilde Q^{(0,2,0)}(\vr) \tilde Q^{(1,0,1)}(\vr)
-  \tilde Q^{(1,1,0)}(\vr) \tilde Q^{(1,1,0)}(\vr) \tilde Q^{(0,0,2)}(\vr) \nonumber \\
&-  \tilde Q^{(2,0,0)}(\vr) \tilde Q^{(0,1,1)}(\vr) \tilde Q^{(0,1,1)}(\vr)
 < 0.  \label{eq:minor3}  
\end{align}
The inequalities \eqref{eq:minor1}, \eqref{eq:minor2}, and \eqref{eq:minor3}, follow directly from 
\begin{align}
 \tilde Q^{(2,0,0)},  \tilde Q^{(0,2,0)},  \tilde Q^{(0,0,2)} \leq -0.842 \kappa^2,
\quad  |\tilde Q^{(1,1,0)}|,  |\tilde Q^{(1,0,1)}|,  |\tilde Q^{(0,1,1)}| \leq 0.2113 \kappa^2. 
\label{eq:tilQbounds}
\end{align}
To prove \eqref{eq:tilQbounds}, we define $Q_R^{(\vn)} \defeq \frac{1}{\kappa^{ \onenorm{\vn} }} \mathrm{Re}(Q^{(\vn)})$, $Q_I^{(\vn)} \defeq \frac{1}{\kappa^{ \onenorm{\vn} }} \mathrm{Im}(Q^{(\vn)})$ and note that 
\[
\frac{1}{\kappa}\tilde Q^{(1,0,0)} = \frac{Q_R^{(1,0,0)}Q_R + Q_I^{(1,0,0)}Q_I  }{|Q|}. 
\]
After simplification, we obtain
\begin{align}
\frac{1}{\kappa^2}
\tilde Q^{(2,0,0)} 
%
&=
\left(1-\frac{Q_R^2}{|Q|^2}\right) \frac{{Q_R^{(1,0,0)}}^2}{|Q|} -\frac{2Q_R Q_R^{(1,0,0)}  Q_I Q_I^{(1,0,0)} + Q_I^2 {Q_I^{(1,0,0)}}^2  }{|Q|^3} \nonumber \\
 &\hspace{0.3cm}+ \frac{{Q_I^{(1,0,0)}}^2   + Q_I Q_I^{(2,0,0)} }{|Q|}+ \frac{Q_R}{|Q|} Q_R^{(2,0,0)},
\label{eq:lhsoftildq20}
\end{align}
and 
\begin{align}
\frac{1}{\kappa^2}
\tilde Q^{(1,1,0)}
%
&= Q_R^{(1,1,0)} \frac{Q_R}{|Q|} + \frac{Q_R^{(1,0,0)}Q_R^{(0,1,0)}}{|Q|} \left(1- \frac{Q_R^2}{|Q|^2} \right) 
+ \frac{ Q_I^{(1,1,0)}Q_I + Q_I^{(1,0,0)}Q_I^{(0,1,0)} }{|Q|} \nonumber \\
&- \frac{Q_R^{(0,1,0)}Q_R Q_I^{(1,0,0)}Q_I +  Q_I^{(0,1,0)}Q_I (Q_R^{(1,0,0)}Q_R + Q_I^{(1,0,0)}Q_I)  }{|Q|^3}. \label{eq:tilQ11}
\end{align}
The upper bounds in \eqref{eq:tilQbounds} now follow by applying the following inequalities with $\epsilon = 0.0005$ to \eqref{eq:lhsoftildq20} and \eqref{eq:tilQ11}:
\begin{align*}
 &Q_R(\vr) \leq \abs{\bar Q(\vr)} + \epsilon \leq 1+\epsilon,   \quad
  Q_R(\vr) \geq \abs{\bar Q(\vr)} - \epsilon \geq 0.9298 - \epsilon, \quad  
\abs{Q_I(\vr)} \leq \bar Q_I(\vr) + \epsilon \leq 0.0702 \\
&|Q^{(1,0,0)}_R(\vr)| \leq 
 |\bar Q_R^{(1,0,0)}(\vr)| + \epsilon \leq 0.3756 +\epsilon,\quad 
|Q^{(1,0,0)}_I(\vr)| \leq 
|\bar Q^{(1,0,0)}_I(\vr)| + \epsilon \leq 0.0541 +\epsilon. \\
&Q_R^{(2,0,0)}(\vr)  \leq  
\bar Q_R^{(2,0,0)}(\vr) + \epsilon \leq -0.8572 + \epsilon, \quad
\abs{Q_I^{(2,0,0)}(\vr)}  \leq  
\abs{\bar Q^{(2,0,0)}(\vr) } + \epsilon \leq -0.1005 +\epsilon \\
&|Q_R^{(1,1,0)}| \leq 
|\bar Q_R^{(1,1,0)}(\vr)| + \epsilon \leq 0.1967 + \epsilon, \quad 
|Q_I^{(1,1,0)}| \leq 
|\bar Q_I^{(1,1,0)}(\vr)| + \epsilon \leq 0.0921 + \epsilon.
\end{align*}
Here, $\bar Q_R^{(\vn)} = \frac{1}{\kappa^{ \onenorm{\vn} }} \mathrm{Re}(\bar Q^{(\vn)})$ and $\bar Q_I^{(\vn)} = \frac{1}{\kappa^{ \onenorm{\vn} }} \mathrm{Im}(\bar Q^{(\vn)})$. 
The inequalities above follow from combing \eqref{eq:conddiffQbarQinfipr} with  Proposition \ref{propcanc2} below using the triangle inequality (recall that $\bar Q^{(\vn)}(\vr)$ is real). 

\begin{proposition}
Suppose that $\vect{0} \in \T$, and let $Q(\vect{0}) = 1$. Then, for $|\vr| \leq \cfar / \N$,
\begin{align*}
  &\abs{\bar Q(\vr) } \leq 1, \quad 
  \bar Q_R(\vr) \geq 0.92982 \\
& \bar Q_R^{(2,0,0)}(\vr)  \leq  -0.8572, \quad 
\abs{\bar Q_I^{(2,0,0)}(\vr) }  \leq  0.1005 \\
&|\bar Q_R^{(1,1,0)}| \leq 0.1967, \quad 
|\bar Q_I^{(1,1,0)}|  \leq 0.0921 \\
&|\bar Q^{(1,0,0)}_R(\vr)|  \leq  0.3756, \quad
|\bar Q^{(1,0,0)}_I(\vr)|  \leq  0.0541.
\end{align*}
\label{propcanc2}
\end{proposition}

The proof of Proposition~\ref{propcanc2} is analogous to corresponding results established for the 1D and 2D case in \cite[Proofs of Lemmas 2.3 and 2.4, Appendix C.2]{candes_towards_2014}, and therefore omitted. 

\end{proof}

\newcommand{\dT}{\mc S} 

\section{\label{sec:proofdiscrete}Proof of discrete results (Theorems \ref{cor:discretesuperresII} and \ref{cor:mimodiscrete})}
\label{discreteProofs}

The following proposition---standard in the theory of compressed sensing (see e.g.,~\cite{candes_robust_2006})---shows that the existence of a certain dual polynomial guarantees that $\mathrm{L1}(\vy)$ in \eqref{eq:l1minmG} succeeds in reconstructing $\vs$. 

\begin{proposition}
Let $\vy= \mA \mF_{\mathrm{grid}} \vs$ and let $\dT$ denote the support of $\vs$. Assume that the columns of $\mA \mF_{\mathrm{grid}}$ indexed by $\dT$ are linearly independent. If there exists a vector $\vv$ in the row space of $\mA \mF_{\mathrm{grid}}$ with 
\begin{align}
\vv_{\dT} = \sign(\vs_{\dT}) \quad \text{and} \quad \infnorm{\vv_{\comp{\dT}}} < 1
\label{eq:dualcerl1}
\end{align}
then $\vs$ is the unique minimizer of $\mathrm{L1}(\vy)$ in \eqref{eq:l1minmG}. 
Here, $\vv_{\dT}$ is the sub-vector of $\vv$ corresponding to the entries indexed by $\dT$. 
\end{proposition}
The proof now follows directly from Lemma~\ref{prop:dualpolynomial}, and from the conditions of Lemma~\ref{prop:dualpolynomial} being satisfied by the assumptions of Theorems \ref{thm:incoherence} and \ref{thm:mainresmimo}, 
as shown in Sections \ref{sec:proofthm1} and \ref{sec:proofmimo}. 
To see this, set $\vu = \sign(\vs_{\dT})$ in Lemma~\ref{prop:dualpolynomial} and consider the polynomial $Q(\vr)$ from Lemma~\ref{prop:dualpolynomial}. Define $\vv$ as
 \mbox{$[\vv]_{(n_1,n_2,n_3)} = Q([n_1/\K_1,n_2/\K_2,n_3/\K_3])$} and note that it satisfies \eqref{eq:dualcerl1} since \linebreak\mbox{$Q([n_1/\K_1,n_2/\K_2,n_3/\K_3]) = \sign(\vs_{(n_1,n_2,n_3)})$} for $\vectbr{n_1,n_2,n_3} \in \dT$ and $\abs{Q([n_1/\K_1,n_2/\K_2,n_3/\K_3])} < 1$ for $\vectbr{n_1,n_2,n_3} \notin \dT$. 


\section{Proof of Lemma \ref{lem:denoisinglasso}}
\label{pfmlem}

\newcommand\optz{\vz}
\newcommand\optmu{\mu}

\newcommand\estmu{\tilde \mu}
\newcommand\estz{\tilde \vz}
\newcommand\errm{\xi} 
\newcommand\proj[1]{P_{#1}}

Lemma~\ref{lem:denoisinglasso} is essentially a consequence of results established in \cite{tang_near_2015}, in particular the following proposition. 

\begin{proposition}[{\cite[Thm.~1]{tang_near_2015}}]
\label{thm:denoisinglinespect}
Let $\optz = \sum_{k=1}^\S b_k \vf(\f_k)$, and suppose the $\f_1,\ldots,\f_\S$ obey the minimum separation condition~\eqref{eq:minsepintro1D}. 
Suppose that the dual norm of $\norm[\setA]{\cdot}$, defined as
$\dualnorm[\setA]{\cdot} \defeq \sup_{\f \in [0,1]} \left| \innerprod{\cdot}{\vf(\f)}\right|$, obeys
\begin{align}
\label{eq:conditionhighprob}
\dualnorm[\setA]{\vw} 
\leq  \frac{\lambda}{\eta}.
\end{align}
Then, for $\eta>1$ sufficiently large, the proximity operator in~\eqref{eq:proxopat} with regularization parameter $\lambda = \eta \sigma \sqrt{4\L \log(\L)}$ applied to $\vy = \optz + \vw$ with $\vw \sim \mc N(0, \sigma^2 \mI)$ obeys
\[
\norm[2]{\estz(\optz + \vw, \lambda) - \optz}^2
\leq
c \sigma^2 \S \log \L.
\]
\end{proposition}

Since~\eqref{eq:conditionhighprob} holds with high probability~\cite[App.~C]{bhaskar_atomic_2012} for our choice of $\lambda$, Theorem~\ref{thm:denoisinglinespect} provides an error bound that holds with high-probability. 
However, we need to characterize the expectation of the estimation error, thus slightly more work is required to establish Lemma~\ref{lem:denoisinglasso}. 

For notational convenience, we set $\estz = \estz(\optz + \vw, \lambda)$ and let $\ve \defeq \optz - \estz$ be the error vector. 
Note that 
\begin{align}
\EX{\norm[2]{\ve}^2}
= 
\EX{\norm[2]{\ve}^2 \ind{ \dualnorm[\setA]{\vw} < \lambda/\eta } }
+
\EX{\norm[2]{\ve}^2 \ind{ \dualnorm[\setA]{\vw} \geq \lambda/\eta } }.
\end{align}
The lemma follows by upper bounding the terms on the RHS. 
The first term is smaller than $c \sigma^2 \S \log \L$ by Proposition~\ref{thm:denoisinglinespect}. 
We next upper bound the second term. 
By the Cauchy-Schwarz inequality, we have
\begin{align}
\EX{\norm[2]{\ve}^2 \ind{ \dualnorm[\setA]{\vw} \geq \lambda/\eta } }
&\leq
\sqrt{ \EX{\norm[2]{\ve}^4} }
\sqrt{ \PR{\dualnorm[\setA]{\vw} \geq \lambda/\eta } } 
\leq 
c \sigma^2 \log(\L),
\end{align}
where the last inequality follows from Lemma~\ref{lem:largenoisebound}, which provides a crude upper bound on $\EX{\norm[2]{\ve}^2}$, and Lemma~\ref{lem:prwbound}, which states an upper bound on the probability for the atomic norm of $\vw$ being large. 
\begin{lemma}
\label{lem:largenoisebound}
\[
\sqrt{\EX{\norm[2]{\ve}^4}}
\leq c \sigma^2 \L \log(\L).
\]
\end{lemma}

\begin{lemma}
\label{lem:prwbound}
\[
\PR{\dualnorm[\setA]{\vw} \geq \lambda/\eta } 
\leq 
\frac{1}{\L^2}.
\]
\end{lemma}
It remains to prove Lemmas~\ref{lem:largenoisebound} and \ref{lem:prwbound}. 


\begin{proof}[Proof of Lemma~\ref{lem:largenoisebound}]
We show below that 
\begin{align}
\label{eq:deterrorb}
\norm[2]{\ve}^2 
\leq 
8\norm[2]{\vw}^2
+ c \lambda^2 \frac{\S}{\L}
=
8\norm[2]{\vw}^2
+ c \eta^2 \sigma^2 \log(\L) \S.
\end{align}
It follows that 
\[
\EX{\norm[2]{\ve}^4}
\leq 
64 \EX{\norm[2]{\vw}^4}
+ 16 \EX{\norm[2]{\vw}^2} \eta^2 \sigma^2 4 \L \log(\L) 
+ c^2 \eta^4 \sigma^4 \L^2 \log^2(\L) 
\leq
c' \eta^4 \sigma^4 \L^2\log^2(\L) .
\]
Here, we used that $\norm[2]{\vw}^2$ is $\chi^2$ distributed and the second moment of a $\chi^2$ random variable with $\L$ degrees of freedom is $\L(\L+2)$. 
This proves the lemma. 

It remains to prove~\eqref{eq:deterrorb}. 
First note that $\optz = \int_0^1 \vf(\f) \optmu(d\f)$, where $\optmu$ is the measure defined as $\optmu(\f) = \sum_{k=1}^\S b_k \delta(\f - \f_k)$, where $\delta$ is the Dirac measure. 
Next, let $\estmu$ be the representing measure of the estimate $\estz$ with minimum total variation norm (for a definition of the total variation norm, denoted by $\norm[TV]{\cdot}$ see~\cite[App.~A]{candes_towards_2014}), and note that $\norm[TV]{\estmu} = \norm[\setA]{\estz}$. 
By optimality of $\estz$, we have 
\[
\frac{1}{2}\norm[2]{\vy - \estz}^2 + \lambda \norm[TV]{\estmu}
\leq
\frac{1}{2}\norm[2]{\vy - \optz}^2 + \lambda \norm[TV]{\optmu}. 
\]
Denote by $\proj{\T}(\estmu)$ the projection of the measure $\estmu$ on the support set of $\optmu$, i.e., $\T = \{\f_1,\ldots,\f_\S\}$, likewise $\proj{\comp{\T}}(\estmu)$ is the projection onto the complement of $\T$. 
With this notation and $\vy = \optz + \vw$, the former inequality is equivalent to 
\begin{align}
\label{eq:boundy}
\frac{1}{2}\norm[2]{\ve}^2 
+ \lambda (\norm[TV]{\proj{\T}(\estmu)} + \norm[TV]{\proj{\comp{\T}}(\estmu)})
\leq \innerprod{\ve}{\vw} + \lambda \norm[TV]{\optmu}.
\end{align}
We next lower bound $\norm[TV]{\proj{\T}(\estmu)}$. 
To this end, note that Proposition 2.1 in \cite{candes_towards_2014} guarantees that there exists a polynomial $Q(\f) = \innerprod{\vq}{\vf(\f)}$ that interpolates the sign pattern of $\optmu$ and obeys $|Q(\f)| \leq 1$. 
Let $\errm \defeq \estmu - \optmu$ be the difference measure. 
Since $Q$ interpolates the sign of $\optmu$,
\[
\norm[TV]{\optmu} + \int_0^1 Q(\f) \proj{\T}(\errm)(d\f) 
\leq
\norm[TV]{\proj{\T}(\estmu)}.
\]
Using this inequality in~\eqref{eq:boundy} and rearranging terms yields
\begin{align*}
\frac{1}{2}\norm[2]{\ve}^2 
+ \lambda \norm[TV]{\proj{\comp{\T}}(\estmu)}
&\leq 
\innerprod{\ve}{\vw} - \lambda \int_0^1 Q(\f) \proj{\T}(\errm)(d\f) \\
&= 
\innerprod{\ve}{\vw} - \lambda \int_0^1 Q(\f) \errm(d\f) 
+
\lambda \int_0^1 Q(\f) \proj{\comp{\T}}(\errm)(d\f) \\
&= 
\innerprod{\ve}{\vw} - \lambda \innerprod{\ve}{\vq} 
+
\lambda \int_0^1 Q(\f) \proj{\comp{\T}}(\errm)(d\f) \\
&\leq 
\innerprod{\ve}{\vw} - \lambda \innerprod{\ve}{\vq} 
+
\lambda \norm[TV]{\proj{\comp{\T}}(\errm) }.
\end{align*}
where the penultimate inequality follows from $Q(\f) = \innerprod{\vq}{\vf(\f)}$ and Parseval's identity, 
and the last inequality holds by $|Q(\f)|\leq 1$ .
We have shown that 
\begin{align}
\label{eq:firstebound}
\frac{1}{2}\norm[2]{\ve}^2 
&\leq 
\innerprod{\ve}{\vw} - \lambda \innerprod{\ve}{\vq}. 
\end{align}
Next, using the inequality $ab \leq \frac{1}{2}(a^2 + b^2)$ with $a = \norm[2]{\ve}/(2\sqrt{\lambda}), b = 2 \sqrt{\lambda}\norm[2]{\vq}$ we obtain 
\[
-\lambda 
\innerprod{\ve}{\vq}
\leq 
\lambda \norm[2]{\ve} \norm[2]{\vq}
\leq 
\frac{1}{8} \norm[2]{\ve}^2  + 2 \lambda^2 \norm[2]{\vq}^2.
\]
Using this in~\eqref{eq:firstebound} yields
\begin{align*}
\frac{3}{8}\norm[2]{\ve}^2 
&\leq 
\innerprod{\ve}{\vw} 
+
2 \lambda^2 \norm[2]{\vq}^2. \\
&\leq 
\norm[2]{\ve} \norm[2]{\vw}
+
2 \lambda^2 \norm[2]{\vq}^2.
\end{align*}
Application of $ab \leq \frac{1}{2}(a^2 + b^2)$ once again yields $\norm[2]{\ve} \norm[2]{\vw} \leq \frac{1}{8}\norm[2]{\ve}^2  + 2 \norm[2]{\vw}^2$ which in turn yields 
\[
\frac{1}{4}\norm[2]{\ve}^2 
\leq 
2 \norm[2]{\vw}^2
+
2 \lambda^2 \norm[2]{\vq}^2.
\]
By Parseval's theorem, and H\"older's inequality we have 
\[
\norm[2]{\vq}^2 
= \innerprod{Q(\f)}{Q(\f)} 
\leq \norm[1]{Q} \norm[\infty]{Q}
\leq \frac{c \S}{\L}, 
\]
where the last inequality follows from $\norm[1]{Q} \leq \frac{c \S}{\L}$ \cite[Lem.~4]{tang_near_2015}. 
This concludes the proof of inequality~\eqref{eq:deterrorb}. 
\end{proof}


\begin{proof}[Proof of Lemma~\ref{lem:prwbound}]
As shown in~\cite[C.4, with $N=n$]{bhaskar_atomic_2012}, a consequence of  Bernstein's polynomial inequality is that 
\[
\dualnorm[\setA]{\vw}
\leq
\left( 1 + 4\pi \right) \max_{m = 0,\ldots,\L-1} 
\underbrace{
\left| \sum_{\ell=0}^{\L-1} w_\ell e^{i2\pi \ell (m/\L) } \right|
}_{Z_m}
.
\]
The random variable $Z_m$ is stochastically upper bounded by a Gaussian random variable with variance $\sigma^2 \L$. 
It follows that, for any $\beta > 1/\sqrt{2\pi}$,  
\[
\PR{ |Z_m| \geq \sigma \sqrt{\L} \beta }
\leq 
2e^{-\beta^2}.
\]
The proof is concluded by noting that, by the union bound
\[
\PR{\dualnorm[\setA]{\vw} \geq \sigma \sqrt{\L} \sqrt{4\log(\L)}}
\leq
2 \L e^{-4 \log(\L)} \leq \frac{1}{\L^2}.
\]
\end{proof}

\end{document}